\documentclass[USenglish]{lmcs}
\pdfoutput=1

\usepackage{lastpage}
\lmcsdoi{18}{3}{19}
\lmcsheading{}{\pageref{LastPage}}{}{}%
{Oct.~01,~2021}{Aug.~09,~2022}{}

\pdfoutput=1

\usepackage[mathscr,scaled=1.15]{urwchancal} 
\usepackage[T1]{fontenc}
\usepackage[latin9]{inputenc}
\usepackage{babel}

\usepackage{graphicx}
\usepackage{url}
\usepackage{color}
\usepackage[unicode=true,
    bookmarks=true,bookmarksnumbered=false,bookmarksopen=false,
    breaklinks=true,pdfborder={0 0 1},colorlinks=true]
    {hyperref}
\usepackage{mathtools}
\usepackage{amssymb}
\usepackage{latexsym} 
\usepackage{stmaryrd} 
\usepackage{centernot}

\newcommand\cancel[1]{\tikz[baseline=0]\draw[red,semithick,line cap=round]node[rectangle,minimum width=0pt,minimum height=0pt,text height=1.7ex,draw=none,text=black!75,node font=\small,anchor=base,inner sep=0pt,behind path](s){$#1$}(s.185)--(s.5);}

\usepackage[misc,geometry]{ifsym}

\usepackage[lined,boxed,commentsnumbered]{algorithm2e}

\usepackage{tikz}
\usetikzlibrary{arrows,automata}
\usetikzlibrary{decorations}
\usetikzlibrary{decorations.pathmorphing}
\usetikzlibrary{arrows.meta}
\usetikzlibrary{shapes.geometric}
\usetikzlibrary{shapes.multipart}
\usetikzlibrary{shapes}

\usepackage{pgf}
\usepackage{pgfplots}
\pgfplotsset{compat=1.16}

\usepackage{wrapfig, tabularx,booktabs}

\usepackage{colortbl}

\usepackage[textsize=tiny]{todonotes}

\usepackage{xstring}

\usepackage{scalerel,stackengine}
\stackMath%
\newcommand{\stretchedhat}[2]{\stackon[-4.6pt]{#2}{\makebox[0pt]{\hstretch{#1}{\widehat{\phantom{\;\;\;\;\;\;\;\;}}}}}}

\SetAlCapSty{textnormal}
\SetArgSty{textnormal}
\SetProgSty{textnormal}
\SetFuncArgSty{textnormal}

\LinesNumbered%
\DontPrintSemicolon%
\SetKwProg{For}{for}{\string:}{}
\SetKwProg{Fn}{def}{\string:}{}
\SetKwProg{While}{while}{\string:}{}
\SetKwIF{If}{ElseIf}{Else}{if}{\string:}{elif}{else\string:}{}%

\global\long\def\defEq{\operatorname{\coloneqq}}%
\global\long\def\varname#1{\mathsf{#1}}%
\global\long\def\ccsStop{\mathbf{0}}%
\global\long\def\ccsPrefix{\ldotp\!}%
\global\long\def\ccsChoice{+}%
\global\long\def\ccs{\mathsf{CCS}}%
\global\long\def\ccsIdentifier#1{\mathsf{#1}}%
\global\long\def\ccsInm#1{#1}%

\global\long\def\rel#1{\mathcal{#1}}%
\global\long\def\bigo{\mathcal{O}}%
\global\long\def\relSize#1{\lvert\mathord{#1}\rvert}%
\newcommand{\powerSet}[1]{\mathbf{2}^{#1}}
\newcommand{\bellNumber}[1]{\mathbf{B}(#1)}

\newcommand{\gameMoveX}[1]{\mathrel{\smash{\xrightarrowtail{#1}}}}%
\newcommand{\gameMove}{\gameMoveX{\hspace*{0.5em}}}%
\newcommand{\game}{\mathcal{G}}%
\newcommand{\gameSpectroscopy}{\mathcal{G}_\triangle}%
\newcommand{\attackerPos}[1]{{(#1)}_\mathtt{a}}
\newcommand{\defenderPos}[1]{{(#1)}_\mathtt{d}}
\newcommand{\partition}[1]{\mathscr{#1}}

\newcommand{\attackerSubscript}{{\operatorname{a}}}
\newcommand{\defenderSubscript}{{\operatorname{d}}}

\newcommand{\proc}{\mathcal{P}}
\newcommand{\system}{\mathcal{S}}%
\newcommand{\action}[1]{\mathit{#1}}
\newcommand{\actions}{\Sigma}%
\newcommand{\step}[1]{\mathrel{\smash{\xrightarrow{#1}}}}%

\newcommand{\hml}{\mathsf{HML}}
\newcommand{\hmlA}{\hml[\actions]}
\newcommand{\hmlObs}[1]{\langle#1\rangle}
\newcommand{\hmlObsI}[1]{\hmlObs{\ccsInm{#1}}}
\newcommand{\hmlAnd}[2]{{\bigwedge_{#1 \in #2}}}
\newcommand{\hmlAndS}{{\bigwedge}}
\newcommand{\hmlTrue}{\mathsf{T}}
\newcommand{\hmlNeg}{\neg}
\newcommand{\hmlSemantics}[3]{{\llbracket #1 \rrbracket}^{#2}_{#3}}
\newcommand{\hmlStrategies}{\mathsf{Strat}}
\newcommand{\hmlPrune}{\mathsf{prune\_dominated}}

\newcommand{\expr}{\mathsf{expr}}
\newcommand{\prices}{\mathbf{Pr}}

\let\obs\undefined
\newcommand{\obs}[1]{\mathcal{O}_{\mathrm{#1}}}
\newcommand{\bEquiv}[1]{\sim_{\mathrm{#1}}}
\newcommand{\bPreord}[1]{\preceq_{\mathrm{#1}}}

\newcommand{\refDef}[1]{Definition~\ref{#1}}
\newcommand{\refExample}[1]{Example~\ref{#1}}
\newcommand{\refRem}[1]{Remark~\ref{#1}}
\newcommand{\refThm}[1]{Theorem~\ref{#1}}
\newcommand{\refLem}[1]{Lemma~\ref{#1}}
\newcommand{\refCor}[1]{Corollary~\ref{#1}}
\newcommand{\refSec}[1]{Section~\ref{#1}}
\newcommand{\refSubsec}[1]{Subsection~\ref{#1}}
\newcommand{\refFig}[1]{Figure~\ref{#1}}
\newcommand{\refClaim}[1]{Claim~\eqref{#1}}

\global\long\def\ie{\text{i.e.\,}}%

\colorlet{darkgreen}{green!60!black}
\colorlet{darkred}{red!50!black}

\makeatletter
\newbox\xrat@below
\newbox\xrat@above
\newcommand{\xrightarrowtail}[2][]{%
  \setbox\xrat@below=\hbox{\ensuremath{\scriptstyle #1}}%
  \setbox\xrat@above=\hbox{\ensuremath{\scriptstyle #2}}%
  \pgfmathsetlengthmacro{\xrat@len}{max(\wd\xrat@below,\wd\xrat@above)+.6em}%
  \mathrel{\tikz [>->,baseline=-.58ex,line width=0.43pt]
                 \draw (0,0) -- node[below=-2pt] {\box\xrat@below}
                                node[above=-2pt] {\box\xrat@above}
                       (\xrat@len,0) ;}} 
\makeatother

\newcommand{\notand}{{\vphantom{\land}\smash{\text{\raisebox{-1pt}{\reflectbox{\rotatebox{90}{$\ngtr$}}}}}}}

\begin{document}
\title[Linear-Time--Branching-Time Spectroscopy]{\texorpdfstring{\vspace*{0.5em}}{}Deciding All Behavioral Equivalences at Once:\texorpdfstring{\\}{ }A Game for Linear-Time--Branching-Time Spectroscopy\texorpdfstring{\rsuper*}{}}%
\titlecomment{{\lsuper*}This is an extended and revised version of ``A Game for Linear-Time--Branching-Time Spectroscopy''~\cite{bisping2021ltbtsTacas}. It contains an important correction to the algorithm and a correctness proof for the revised algorithm.}
\author[B.~Bisping]{Benjamin Bisping\lmcsorcid{0000-0002-0637-0171}}[a]
\author[D.N.~Jansen]{David N. Jansen\lmcsorcid{0000-0002-6636-3301}}[b,c]
\author[U.~Nestmann]{Uwe Nestmann\lmcsorcid{0000-0002-8520-5448}}[a]
\address{Technische Universit\"at Berlin, Berlin, Germany}
\email{\{benjamin.bisping, uwe.nestmann\}@tu-berlin.de}
\address{State Key Laboratory of Computer Science, Institute of Software, Chinese Academy of Sciences, Beijing, China}
\email{dnjansen@ios.ac.cn}
\address{University of Chinese Academy of Sciences, Beijing, China}

\keywords{
  Process equivalence spectrum,
  Distinguishing formulas,
  Bisimulation games.}
\begin{abstract}

  We introduce a generalization of the bisimulation game
  that finds distinguishing Hennessy--Milner logic formulas
  from every finitary, subformula-closed language in van Glabbeek's linear-time--branching-time spectrum
  between two finite-state processes.
  We identify the relevant dimensions that measure expressive power
  to yield formulas belonging to the coarsest distinguishing behavioral preorders and equivalences;
  the compared processes are equivalent in each coarser behavioral equivalence from the spectrum.
  We prove that the induced algorithm can determine the best fit of (in)equivalences for a pair of processes.

\end{abstract}
\maketitle
\section{Introduction}

Have you ever looked at two system models and wondered what would be the finest notions of behavioral equivalence to equate them---or, conversely:
the coarsest ones to distinguish them?
We often run into this situation when analyzing models and, especially, when devising examples for teaching.
We then find ourselves fiddling around with whiteboards and various tools, each implementing different equivalence checkers.
Would it not be nice to \emph{decide all equivalences at once?}

\begin{figure}[t]
  \begin{tikzpicture}[->,auto,node distance=2cm, rel/.style={dashed,font=\it, blue}, ext/.style={line width=1pt}]
    \node (P1){$\ccsIdentifier{P_1}$};
    \node (P2) [right of=P1, node distance=5cm] {$\ccsIdentifier{P_2}$};

    \node (BC) [below left of=P1] {$\ccsInm{b}\ccsChoice\ccsInm{c}$};
    \node (D) [below right of=P1] {$\ccsInm{d}$};

    \node (BD) [below left of=P2] {$\ccsInm{b}\ccsChoice\ccsInm{d}$};
    \node (CD) [below right of=P2] {$\ccsInm{c}\ccsChoice\ccsInm{d}$};

    \node (Stop1) [below right of=BC] {$\ccsStop$};
    \node (Stop2) [below right of=BD] {$\ccsStop$};

    \path
    (P1) edge[swap] node {$\action{a}$} (BC)
    (P1) edge node {$\action{a}$} (D)
    (BC) edge[swap] node {$\action{b}, \action{c}$} (Stop1)
    (D) edge node {$\action{d}$} (Stop1)
    ;

    \path
    (P2) edge[swap] node {$\action{a}$} (BD)
    (P2) edge node {$\action{a}$} (CD)
    (BD) edge[swap] node {$\action{b}, \action{d}$} (Stop2)
    (CD) edge node {$\action{c}, \action{d}$} (Stop2)
    ;
  \end{tikzpicture}
  \caption{Processes for \refExample{ex:ccs-p1p2}.}%
  \label{fig:ts-example}
\end{figure}

\begin{exa}%
  \label{ex:ccs-p1p2}
  Consider the $\ccs$ process $\ccsIdentifier{P_1} = \ccsInm{a}\ccsPrefix(\ccsInm{b}\ccsChoice\ccsInm{c})\ccsChoice\ccsInm{a}\ccsPrefix\ccsInm{d}$,
  shown in \refFig{fig:ts-example}.
  It describes a machine that can be activated ($\ccsInm a$) and then either is in a state where one can choose from $\ccsInm{b}$ and \ccsInm{c} or where it can only be deactivated again ($\ccsInm d$).
  $\ccsIdentifier{P_1}$ shares a lot of properties with $\ccsIdentifier{P_2} = \ccsInm{a}\ccsPrefix(\ccsInm{b}\ccsChoice\ccsInm{d})\ccsChoice\ccsInm{a}\ccsPrefix(\ccsInm{c}\ccsChoice\ccsInm{d})$.
  For example, they have the same traces (and the same completed traces).
  Thus, they are \emph{(completed) trace equivalent.}

  But they also have differences.
  For instance, $\ccsIdentifier{P_1}$ has a run where it executes $\ccsInm{a}$
  and then cannot do $\ccsInm{d}$,
  while $\ccsIdentifier{P_2}$ does not have such a run.
  Hence, they are \emph{not failure equivalent.}
  Moreover, $\ccsIdentifier{P_1}$ may perform $\ccsInm{a}$ and then choose from $\ccsInm{b}$ and $\ccsInm{c}$,
  and $\ccsIdentifier{P_2}$ cannot.
  This renders the two processes also \emph{not simulation equivalent.}
  Failure equivalence and simulation equivalence are incomparable---that is, neither one follows from the other one.
  \emph{Both} are maximally coarse ways of telling the processes apart.
  Other inequivalences, like bisimulation inequivalence, are implied by both.
  Together, they make (completed) trace equivalence the finest notion to equate the processes.
\end{exa}

In the following, we present a uniform game-based way of finding the most fitting notions of (in)equivalence for process pairs like in \refExample{ex:ccs-p1p2}.

Our approach is based on the fact that notions of process equivalence can be characterized by two-player games.
The defender's winning region in the game corresponds to pairs of equivalent states,
and the attacker's winning strategies correspond to distinguishing formulas in Hennessy--Milner logic (HML).

Each notion of equivalence in van Glabbeek's famous linear-time--branch\-ing-time spectrum~\cite{glabbeek1990ltbt1} can be characterized by a subset of HML with specific distinguishing power.
Some of the notions are incomparable.
So, often a process pair that is equivalent with respect to one equivalence is distinguished by a set of slightly coarser or incomparable equivalences,
without any one of them alone being \emph{the} coarsest way to distinguish the pair.
As with the spectrum of light where a mix of wave lengths shows to us as a color,
there is a ``mix'' of distinguishing capabilities involved in establishing whether a specific equivalence is finest.
We present an algorithm that is meant to analyze what is in the mix.

\subsection*{Contributions} More precisely, this paper makes the following contributions:
\begin{itemize}
\item We \emph{rechart the linear-time--branching-time spectrum of observation languages using ``formula prices''} that capture six dimensions of expressive capabilities used in HML formulas (\refSubsec{subsec:priced-formulas}).
\item We introduce a \emph{special bisimulation game} that neatly characterizes the distinguishing formulas of HML for pairs of states in finite transition systems (\refSubsec{subsec:spectroscopy-game}).
\item We show how to \emph{enumerate the relevant distinguishing formulas} using the attacker's winning region in this game (\refSubsec{subsec:building-formulas}).
\item We define an algorithm that constructs a \emph{finite set of distinguishing formulas guaranteed to contain observations of the weakest possible observation languages,}
  which can be seen as a ``spectroscopy'' of the differences between two processes
  (\refSubsec{subsec:cheapest-formulas}).
\item We present a small \emph{web tool that is able to run the algorithm on finite-state processes} and output a visual representation of the results (\refSubsec{subsec:webtool}). We also report on the distinctions it finds for all the finitary
examples from the report version of the linear-time--branching-time spectrum~\cite{glabbeek2001ltbtsiReport}.
\item Additionally, we quickly report on a \emph{browser computer game based on the spectroscopy mechanics} (\refSubsec{subsec:spectro-browser-game}).
\end{itemize}
We frame the contributions by a roundtrip through the basics of HML, games and the spectrum (\refSec{sec:preliminaries}), a discussion of related work (\refSec{sec:related}), and concluding remarks on future lines of research (\refSec{sec:conclusion}).

Compared to the original conference version of the paper~\cite{bisping2021ltbtsTacas}, this paper includes an important correction of the spectroscopy game in \refSubsec{subsec:spectroscopy-game}. This allows us to now prove correctness of the algorithm in \refSubsec{subsec:algo-correctness}. Why the change has been necessary is discussed in \refSubsec{subsec:difference-tacas}. On the road to establishing the new correctness result, the journal version pays more attention to the fine points of how we assign formulas expressiveness prices. We are also glad to devote more space to presenting more examples, proofs and subalgorithms, and to also cover the enabledness preorder.

\section{Preliminaries: HML, Games, and the Spectrum}%
\label{sec:preliminaries}

We use the concepts of transition systems, games, observations, and notions of equivalence, largely due to the wake of Hennessy and Milner's seminal paper~\cite{hm1980hml}.

\subsection{Transition Systems and Hennessy--Milner Logic}%
\label{subsec:LTS-HML}

\emph{Labeled transition systems} capture a discrete world view, where there is a current state and a branching structure of possible state changes to future states.
\begin{defi}[Labeled transition system]%
  \label{def:transition-system}
  A \emph{labeled transition system} (LTS) is a triple $(\proc,\actions,\mathord{\step{}})$
  where $\proc$ is the set of \emph{states,}
  $\actions$ the set of \emph{actions,}
  and $\mathord{\step{}}\subseteq \proc\times\actions\times \proc$ the \emph{transition relation.}
\end{defi}
Figure~\ref{fig:ts-example} displays the labeled transition system for \refExample{ex:ccs-p1p2}.
In our examples, we use a small fragment of $\ccs$ to describe certain transition system states.
This fragment is called $\mathsf{BCCSP}$ in~\cite{glabbeek2001ltbtsiReport}.
It contains just \emph{action prefixing} ($\ccsInm{a} \ccsPrefix P$),
\emph{summation} ($P_1 \ccsChoice P_2$),
and the \emph{completed process} ($\ccsStop$).
(Note that LTSs are more powerful than the simple examples;
in particular, our tool is also able to handle recursively defined $\ccs$ terms.)
This fragment only needs two rules for its semantics:
\begin{enumerate}
  \item $\ccsInm{a} \ccsPrefix P \; \step{\action{a}} \; P$ \quad for $P \in \ccs$ and $\action{a} \in \actions$, and
  \item $P_1 \ccsChoice P_2 \; \step{a} \; P_i^\prime$ \quad if there is $i \in \{1,2\}$ such that $P_i \step{a} P_i^\prime$.
\end{enumerate}
Silently assuming commutativity and associativity, we sometimes write $P_1 \ccsChoice P_2 \ccsChoice \cdots \ccsChoice P_n$.

\emph{Hennessy--Milner logic}~\cite{hm1980hml} describes \emph{observations} (or ``tests'') 
on such a system.

\begin{defi}[Hennessy--Milner logic]%
  \label{def:hml}
  Given an alphabet $\actions$, the syntax of \emph{Hennessy--Milner logic} formulas, $\hmlA$, is inductively defined as follows:
  \begin{description}[leftmargin=!,labelwidth=\widthof{\bfseries Conjunctions:}]
    \item[Observations] If $\varphi \in \hmlA$ and $a \in \actions$, then $\hmlObs{a} \varphi \in \hmlA$.
    \item[Conjunctions] If $\varphi_i \in \hmlA$ for all $i$ from an index set $I$, then $\hmlAnd{i}{I}\varphi_i \in \hmlA$.
    \item[Negations] If $\varphi \in \hmlA$, then $\hmlNeg \varphi \in \hmlA$.
  \end{description}
\end{defi}
\noindent
Intuitively, $\hmlObs{\action a}\varphi$ means that one can observe a system transition labeled by $\action a$ and then continue to make observation(s) $\varphi$. Conjunction and negation work as known from propositional logic. So, $\hmlObs{a}\hmlNeg\hmlObs{d}\hmlTrue$ can be read as ``one can observe $a$ in such a way that afterwards one cannot observe a $d$.'' We will provide a common game semantics for $\hml$ in the following subsection.

We often write $\hmlAndS\{\varphi_0, \varphi_1, \dots\}$ for $\hmlAnd{i}{I}\varphi_i$.
$\hmlTrue$ denotes $\hmlAndS \varnothing$, the nil-element of the syntax tree,
and $\hmlObs{a}$ is a short-hand for $\hmlObs{a}\hmlTrue$.
We also implicitly assume that formulas are flattened
in the sense that conjunctions do not contain other conjunctions as immediate subformulas.
We will sometimes talk about the syntax tree height of a formula and consider the height of $\hmlTrue$ to equal $0$.

In principle, \refDef{def:hml} can also be read to be infinitary with respect to branching degree or recursion depth.
Allowing infinite index sets $I$ enables formulas like $\hmlAnd{n}{\mathbb{N}}{\hmlObs{a}^n}$, ``one can observe an arbitrary number of $a$s.''
A coinductive reading makes formulas like $\hmlObs{a}^\omega$, ``one can observe an infinite sequence of $a$s,'' possible.
Note that the formulas both have infinite height.
Obviously, this might add a lot of expressiveness to $\hml$.
For the scope of this paper, we are concerned with finite formulas only.
The prices we are going to define do not distinguish properly between different infinitely branching or infinitely deep formulas.

\subsection{Game Semantics of HML\label{subsec:hml-game}}

Let us fix some notions for \emph{Gale--Stewart-style reachability games} where the defender wins all infinite plays.

\begin{defi}[Games]%
  \label{def:simple-game}
  A \emph{reachability game} $\game[g_{0}] = (G, G_\defenderSubscript, \gameMove, g_0)$ is played on a directed graph consisting of
  \begin{itemize}
    \item a set of \emph{game positions} (vertices) $G$, partitioned into
    \begin{itemize}
      \item a set of \emph{defender positions} $G_\defenderSubscript\subseteq G$
      \item a set of \emph{attacker positions} $G_\attackerSubscript\defEq G \setminus G_\defenderSubscript$,
    \end{itemize}
    \item a set of \emph{game moves} (edges) $\operatorname{\gameMove}\subseteq G\times G$, and
    \item an \emph{initial position} $g_{0}\in G$.
  \end{itemize}
\end{defi}

\begin{defi}[Plays and wins]%
  \label{def:plays-wins}
  We call the paths $g_{0}g_{1}\ldots\in G^{\infty}$ with $g_{i}\gameMove g_{i+1}$ \emph{plays} of $\game[g_{0}]$. They may be finite or infinite. The defender \emph{wins} infinite plays. If a finite play $g_{0}\dots g_{n}\!\centernot{\gameMove}$ is stuck, the stuck player loses: The defender wins if $g_{n}\in G_\attackerSubscript$, and the attacker wins if $g_{n}\in G_\defenderSubscript$.
\end{defi}
\begin{defi}[Strategies and winning strategies]%
  \label{def:strategies}
  A (positional, non-deterministic) \emph{attacker strategy} is a subset of the moves starting in attacker states,
  $F \subseteq (G_\attackerSubscript \times G) \cap \mathord{\gameMove}$.
  Similarly, a \emph{defender strategy} is a subset of the moves starting in defender states,
  $F \subseteq (G_\defenderSubscript \times G) \cap \mathord{\gameMove}$.
  If (fairly) picking elements of strategy $F$ ensures a player to win,
  $F$ is called a \emph{winning strategy} for this player.
  The player with a winning strategy for $\game[g_{0}]$ is said to \emph{win} $\game[g_{0}]$.
  If $F$ is a function on $G_\attackerSubscript$ or $G_\defenderSubscript$, respectively, we call it a \emph{deterministic} strategy.
\end{defi}
The games we use in this paper essentially are parity games only colored by 0. So they are positionally determined. This means, for each possible initial position, exactly one of the two players has a positional deterministic winning strategy $F$. We call this partitioning of the game positions the winning regions.

\begin{defi}[Winning regions]%
  \label{def:winning-region}
  The set $W_\attackerSubscript \subseteq G$ of all positions $g$ where the attacker wins $\game[g]$ is called the \emph{attacker winning region.}
  (The defender winning region $W_\defenderSubscript$ is defined analogously.)
\end{defi}

It is well-known that winning regions of \emph{finite} reachability games can be computed in linear time of the number of game moves. (An algorithm for this is discussed in \refSubsec{subsec:deciding-game}.) This is why the \emph{spectroscopy game} that we introduce in \refSubsec{subsec:spectroscopy-game} can easily be used in algorithms. It derives from the following semantics game for HML, where the defender tries to prove a formula and the attacker tries to falsify it.

\begin{defi}[$\hml$ game]%
  \label{def:hml-game}
  For a transition system $\system=(\proc,\actions,\step{})$,
  the \emph{$\hml$ game} $\game_{\hml}^\system[g_{0}]=(G,G_\defenderSubscript,\gameMove,g_{0})$ is played on $G = \proc \times \hmlA$,
  where the defender controls observations and negated conjunctions, that is $(p, \hmlObs{a}\varphi) \in G_\defenderSubscript$ and $(p,\hmlNeg\hmlAnd{i}{I}\varphi_i) \in G_\defenderSubscript$ (for all $\varphi,p,I$), and the attacker controls the rest.

  \medskip
  \noindent
  \begin{tabularx}{\linewidth}{@{}c@{\hspace{\labelsep}}c@{\hspace{\labelsep}}l@{\;}c@{\;}l@{\quad}X@{}}
  \textbullet & \multicolumn{5}{@{}l}{The defender can perform the moves:} \\
  & \textbf{--} & $(p, \hmlObs{a}\varphi)$                       & $\gameMove$ & $(p^\prime, \varphi)$        & if $p \step{a} p^\prime$ and \\
  & \textbf{--} & $(p, \hmlNeg{\hmlAnd{i}{I}\varphi_i})$         & $\gameMove$ & $(p, \hmlNeg\varphi_i)$      & with $i \in I$; \\
  \end{tabularx} \\ \pagebreak[0]
  \begin{tabularx}{\linewidth}{@{}c@{\hspace{\labelsep}}c@{\hspace{\labelsep}}l@{\;}c@{\;}l@{\quad}X@{}}
  \textbullet & \multicolumn{5}{@{}l}{and the attacker can move:} \\
  & \textbf{--} & $(p, \hmlNeg\hmlObs{a}\varphi)$                & $\gameMove$ & $(p^\prime, \hmlNeg\varphi)$ & if $p \step{a} p^\prime$, \\
  & \textbf{--} & $(p, \hmlAnd{i}{I}\varphi_i)\phantom{\hmlNeg}$ & $\gameMove$ & $(p, \varphi_i)$             & with $i \in I$, and \\
  & \textbf{--} & $(p, \hmlNeg\hmlNeg\varphi)$                   & $\gameMove$ & $(p, \varphi)$.
  \end{tabularx}
\end{defi}

\noindent
Like in other logical games in the Ehrenfeucht--Fra{\"{\i}}ss{\'e} tradition, the attacker plays the conjunctions and universal quantifiers, whereas the defender plays the disjunctions and existential quantifiers. For instance, $(p, \hmlObs{a}\varphi)$ is declared as defender position, since $\hmlObs{a}\varphi$ is meant to become true precisely if \emph{there exists} a state $p^\prime$ reachable $p \step{a} p^\prime$ where $\varphi$ is true.

As every move strictly reduces the height of the formula, the game must be finite-depth (and cycle-free) for finite-height formulas, and, for image-finite systems and formulas, also finite. It is determined and the following semantics is total.

\begin{defi}[$\hml$ semantics]%
  \label{def:hml-semantics}
  For a transition system $\system$, 
  the \emph{semantics of $\hml$} is given by defining that $\varphi$ is true at $p$ in $\system$, written $\hmlSemantics{\varphi}{\system}{p}$, iff the defender wins $\game_\hml^\system[(p, \varphi)]$.
\end{defi}

\begin{exa}
  Continuing \refExample{ex:ccs-p1p2}, $\hmlSemantics{\hmlObs{a}\hmlNeg\hmlObs{d}}{\ccs}{\ccsIdentifier{P_2}}$ is false:
  No matter whether the defender plays to
  $(\ccsInm{b}\ccsChoice\ccsInm{d}, \hmlNeg\hmlObs{d})$ or to
  $(\ccsInm{c}\ccsChoice\ccsInm{d}, \hmlNeg\hmlObs{d})$,
  the attacker wins by moving to the stuck defender position $(\ccsStop, \hmlNeg\hmlTrue)$. (Recall that $\hmlTrue$ is the empty conjunction.)
\end{exa}

\subsection{The Spectrum of Behavioral Equivalences\label{subsec:notions-equivalence}}

For different theoretical and practical applications, a universe of notions of behavioral equivalence has been developed. Those equivalences are often defined in terms of relations on transition system state spaces or sets of executions. For the purpose of our paper, we rather focus on the correspondence between $\hml$ observation languages and notions of behavioral equivalence. In this framework, equivalence of processes means that the same observations are true for them.

\begin{figure}[t]
  \begin{center}
\begin{tikzpicture}[auto,node distance=2.1cm]

\node (B){bisimulation};
\node (2S)[below of=B, node distance=1.2cm]{2-nested simulation};
\node (RS)[below left of=2S]{ready simulation};
\node (RT)[below right of=RS]{ready trace};
\node (PF)[right of=RT,node distance=2.6cm]{possible futures};
\node (FT)[below left of=RT]{failure trace\vphantom{p}};
\node (R)[right of=FT]{readiness\vphantom{p}};
\node (IF)[right of=R, node distance=2.8cm]{impossible futures};
\node (S)[left of=FT, node distance=2.6cm]{simulation\vphantom{p}};
\node (F)[below right of=FT]{failures};
\node (T)[below of=F, node distance=1.2cm]{traces};
\node (E)[below of=T, node distance=1.2cm]{enabledness};

\path
(B) edge (2S)
(2S) edge (RS)
(2S) edge (PF)
(RS) edge (RT)
(RS) edge (S)
(RT) edge (FT)
(RT) edge (R)
(PF) edge (R)
(PF) edge (IF)
(S) edge (T)
(FT) edge (F)
(R) edge (F)
(IF) edge (F)
(F) edge (T)
(T) edge (E)
;

\end{tikzpicture}
  \end{center}
\caption{Hierarchy of equivalences/preorders becoming finer towards the top.\label{fig:ltbt-spectrum} Lines mean inclusion of observation languages from top to bottom.}
\end{figure}

\begin{defi}[Distinguishing formula]%
  \label{def:hml-distinguishing}
  A formula $\varphi$ \emph{distinguishes} state $p$ from $q$ iff $\hmlSemantics{\varphi}{}{p}$ is true and $\hmlSemantics{\varphi}{}{q}$ is not.\footnote{Here, and in the following, we usually leave the transition system $\system$ implicit.}
\end{defi}

\begin{exa}%
  \label{exa:distinguish-P1-from-P2}
  The formula $\hmlObs{a}\hmlNeg\hmlObs{d}$ distinguishes $\ccsIdentifier{P_1}$ from $\ccsIdentifier{P_2}$ in \refExample{ex:ccs-p1p2} (but not the other way around).
  The formula $\hmlObs{a}\hmlAndS\{\hmlObs{b}, \hmlObs{d}\}$ distinguishes $\ccsIdentifier{P_2}$ from $\ccsIdentifier{P_1}$.
\end{exa}

\begin{defi}[Observational preorders and equivalences]%
  \label{def:observation-preorder}
  A set of observations, $\obs{\mathit X} \subseteq \hmlA$, \emph{preorders} two states $p,q$, written $p \bPreord{\mathit X} q$, iff no formula $\varphi\in\obs{\mathit X}$ distinguishes $p$ from $q$. If $p \bPreord{\mathit X} q$ and $q \bPreord{\mathit X} p$, then the two are $X$-equivalent, written $p \bEquiv{\mathit X} q$.
\end{defi}

\begin{exa}[Enabledness preorder and equivalence]\label{ex:enabledness}
  $\obs{E} = \{ \hmlObs{a} \mid a \in \actions \} \cup \{ \hmlTrue \}$ defines the preorder on what actions are initially enabled at two compared processes $\bPreord{E}$, respectively the equivalence $\bEquiv{E}$.
  Effectively, it means that $p \bPreord{E} q$ iff, for all $p \step{a} p'$, there is a $q'$ such that $q \step{a} q'$.
  (The $\{ \hmlTrue \}$ is there in order to ensure compatibility with upcoming definitions.)
\end{exa}

The enabledness equivalence $\bEquiv{E}$ is presumably the coarsest equivalence that one may encounter in the literature. There are many noteworthy finer equivalences. A broad overview is given in Figure~\ref{fig:ltbt-spectrum}.

\begin{defi}[Linear-time--branching-time observation languages~\cite{glabbeek2001ltbtsiReport}]%
  \label{def:ltbts}
  The linear-time--branching-time spectrum is a lattice of observation languages (and of entailed process preorders and equivalences). Every observation language $\obs{\mathit X}$ can perform trace observations, that is, $\hmlTrue \in \obs{\mathit X}$ and, if $\varphi \in \obs{\mathit X}$, then $\hmlObs{a}\varphi \in \obs{\mathit X}$. At the more linear-time side of the spectrum we have:

  \medskip

  \noindent
  \begin{tabularx}{\linewidth}{@{}c@{\hspace{\labelsep}}l@{ }l@{ }l@{\;}r@{ }X@{}}
    \textbullet & \emph{trace observations}         & $\obs{T }$: & \multicolumn{3}{@{}l}{Just trace observations,} \\
    \textbullet & \emph{failure observations}       & $\obs{F }$: & $\hmlAnd{i}{I} \hmlNeg\hmlObs{a_i} \in \obs{F}$, \\
    \textbullet & \emph{readiness observations}     & $\obs{R }$: & $\hmlAnd{i}{I} \varphi_i \in \obs{R}$         & if & all $\varphi_i \in \{ \hmlObs{a}, \hmlNeg\hmlObs{a} \mid a \in \actions\}$, \\
    \textbullet & \emph{failure trace observations} & $\obs{FT}$: & $\hmlAnd{i}{I} \varphi_i \in \obs{FT}$        & if & $\varphi_0 \in \obs{FT}$ and $\varphi_i = \hmlNeg\hmlObs{a_i}$ for $i \neq 0$, \\
    \textbullet & \emph{ready trace observations}   & $\obs{RT}$: & $\hmlAnd{i}{I} \varphi_i \in \obs{RT}$        & if & $\varphi_0 \in \obs{RT}$ and $\varphi_i \in \{ \hmlObs{a}, \hmlNeg\hmlObs{a} \mid a \in \actions\}$ for $i \neq 0$, \\
    \textbullet & \emph{impossible futures}         & $\obs{IF}$: & $\hmlAnd{i}{I} \hmlNeg\varphi_i \in \obs{IF}$ & if & all $\varphi_i \in \obs{T}$, and \\
    \textbullet & \emph{possible futures}           & $\obs{PF}$: & $\hmlAnd{i}{I} \varphi_i \in \obs{PF}$        & if & all $\varphi_i \in \{ \psi, \hmlNeg\psi \mid \psi \in \obs{T}\}$.
  \end{tabularx}

  \medskip

  \noindent
  At the more branching-time side, we have simulation observations.
  Every simulation observation language $\obs{\mathit{X}S}$ permits trace observation construction and has full conjunctive capacity,
  that is, if $\varphi_i \in \obs{\mathit{X}S}$ for all $i \in I$,
  then $\hmlAnd{i}{I} \varphi_i \in \obs{\mathit{X}S}$.

  \medskip

  \noindent
  \begin{tabularx}{\linewidth}{@{}c@{\hspace{\labelsep}}l@{ }l@{ }X@{}}
    \textbullet & \emph{simulation observations}            & $\obs{1S}$:          & Just simulation (and trace) observations, \\
    \textbullet & \emph{$n$-nested simulation observations} & $\obs{\mathit{n}S}$: & $\hmlNeg\varphi \in \obs{\mathit{n}S}$ if $\varphi \in \obs{(\mathit{n}-1)S}$, \\
    \textbullet & \emph{ready simulation observations}      & $\obs{RS}$:          & $\hmlNeg\hmlObs{a} \in \obs{RS}$, and \\
    \textbullet & \emph{bisimulation observations}          & $\obs{B}$:           & The same as $\bigcup_{n \in \mathbb{N}} \obs{\mathit{n}S}$, which is exactly $\hmlA$.
  \end{tabularx}
\end{defi}

\noindent
The observation languages of the spectrum differ in how many of the syntactic features of $\hml$ one will encounter when descending into a formula's syntax tree. We will come back to this in \refSubsec{subsec:priced-formulas}.

The languages of \refDef{def:ltbts} and $\obs{E}$ can be ordered by subset relations between them.
The resulting inclusion structure is depicted in Figure~\ref{fig:ltbt-spectrum}.
For two observation languages $\obs{\mathit X}$ and $\obs{\mathit Y}$ with $\obs{\mathit X} \subseteq \obs{\mathit Y}$,
it is clear that $\obs{\mathit X}$ has at most as much distinctive capability as $\obs{\mathit Y}$,
and thus $p \bPreord{\mathit Y} q$ implies $p \bPreord{\mathit X} q$.
(Pay attention to the subset relation and the implication running in opposite directions!)
Thus, the equivalences referred to in Figure~\ref{fig:ltbt-spectrum} imply one-another downwards.
Those implications usually are strict, albeit not for every transition system.

Note that we consider $\hmlAndS \{ \varphi \}$ to be an alias for $\varphi$.
With this aliasing, all the listed observation languages are \emph{closed}
in the sense that all subformulas and partial conjunctions (conjunctions of subsets) within each observation are themselves part of that language.

\begin{defi}[Closed observation language] We say that an observation language $\obs{\mathit X}$ is \emph{closed} if:
\begin{itemize}
  \item $\hmlObs{a}\varphi \in \obs{\mathit X}$ implies $\varphi \in \obs{\mathit X}$,
  \item $\hmlAndS \Phi \in \obs{\mathit X}$ implies $\Phi \subseteq \obs{\mathit X}$ and $\hmlAndS \Phi' \in \obs{\mathit X}$ for every $\Phi' \subseteq \Phi$, and
  \item $(\hmlNeg\varphi) \in \obs{\mathit X}$ implies $\varphi \in \obs{\mathit X}$.
\end{itemize}
\end{defi}
\begin{prop}\label{prop:languages-closed}
  The languages of \refDef{def:ltbts} and $\obs{E}$ of \refExample{ex:enabledness} are closed.
\end{prop}
\begin{proof}
  This can be seen by examining the definition. The aliasing is necessary for all cases with negations under conjunctions. For instance, for failure observations $\obs{F}$, we have to prove that $\hmlAnd{i}{I} \hmlNeg\hmlObs{a_i} \in \obs{F}$ implies $\hmlObs{a_i} \in \obs{F}$ and $\hmlNeg\hmlObs{a_i} \in \obs{F}$  for all $i \in I$, as well as $\hmlAnd{i}{I'} \hmlNeg\hmlObs{a_i} \in \obs{F}$ with $I' \subseteq I$. We easily see $\hmlObs{a_i} \in \obs{F}$ because $\hmlObs{a_i} \in \obs{T} \subseteq \obs{F}$. The second term is not mentioned by the definition, but has $\hmlAndS \{ \hmlNeg\hmlObs{a_i} \}$ as an alias. The alias is mentioned by the definition, as are all $\hmlAnd{i}{I'} \hmlNeg\hmlObs{a_i} \in \obs{F}$ with $I' \subseteq I$. The other observation languages are equally immediate.
\end{proof}

The languages of \refDef{def:ltbts} and $\obs{E}$ thus are \emph{inductive}
in the sense that all observations with finite syntax tree height must be built from smaller observations of the same language.
This is convenient in proofs by structural induction.

\begin{figure}
  \begin{tikzpicture}[->,auto,node distance=1.6cm, rel/.style={dashed,font=\it, blue}, ext/.style={line width=1pt}]
    \node (P3){$\ccsIdentifier{P_3}$};
    \node (P4) [right of=P3, node distance=6cm] {$\ccsIdentifier{P_4}$};

    \node (B+CD) [below left of=P3] {$\ccsInm{b}\ccsChoice\ccsInm{c}\ccsPrefix\ccsInm{d}$};
    \node (F+CE) [below right of=P3] {$\smash[b]{\ccsInm{f}}\ccsChoice\ccsInm{c}\ccsPrefix\ccsInm{e}$};

    \node (B+CE) [below left of=P4] {$\ccsInm{b}\ccsChoice\ccsInm{c}\ccsPrefix\ccsInm{e}$};
    \node (F+CD) [below right of=P4] {$\smash[b]{\ccsInm{f}}\ccsChoice\ccsInm{c}\ccsPrefix\ccsInm{d}$};

    \node (D1) [below of=B+CD] {$\ccsInm{d}$};
    \node (E1) [below of=F+CE] {$\ccsInm{e\vphantom{d}}$};

    \node (E2) [below of=B+CE] {$\ccsInm{e\vphantom{d}}$};
    \node (D2) [below of=F+CD] {$\ccsInm{d}$};

    \node (Stop1) [circle,below right of=D1] {$\ccsStop$};
    \node (Stop2) [circle,below left of=D2] {$\ccsStop$};

    \path
    (P3) edge[swap] node {$\action{a}$} (B+CD)
    (P3) edge node {$\action{a}$} (F+CE)
    (B+CD) edge[shorten >=0.4pt] node {$\action{b}$} (Stop1.102)
    (B+CD) edge[swap] node {$\action{c}$} (D1)
    (D1) edge[swap] node {$\action{d}$} (Stop1)
    (F+CE) edge[swap,shorten >=0.4pt] node {$\smash[b]{\action{f}}$} (Stop1.78)
    (F+CE) edge node {$\action{c}$} (E1)
    (E1) edge node {$\action{e\vphantom{d}}$} (Stop1)
    ;

    \path
    (P4) edge[swap] node {$\action{a}$} (B+CE)
    (P4) edge node {$\action{a}$} (F+CD)
    (B+CE) edge[shorten >=0.4pt] node {$\action{b}$} (Stop2.102)
    (B+CE) edge[swap] node {$\action{c}$} (E2)
    (E2) edge[swap] node {$\action{e\vphantom{d}}$} (Stop2)
    (F+CD) edge[swap,shorten >=0.4pt] node {$\smash[b]{\action{f}}$} (Stop2.78)
    (F+CD) edge node {$\action{c}$} (D2)
    (D2) edge node {$\action{d}$} (Stop2)
    ;
  \end{tikzpicture}
  \caption{Additional processes for \refExample{exa:distinguishing-formulas}.
      (They differ by the position of $\ccsInm{d}$ and $\ccsInm{e}$.)}%
  \label{fig:ts-example-P3P4}
\end{figure}

\begin{exa}\label{exa:distinguishing-formulas}
  As explained in Examples~\ref{ex:ccs-p1p2} and~\ref{exa:distinguish-P1-from-P2},
  process $\ccsIdentifier{P_1}$ has a run where it executes $\ccsInm{a}$ and then cannot do $\ccsInm{d}$,
  while $\ccsIdentifier{P_2}$ does not have such a run,
  so $\ccsIdentifier{P_1} \not\bPreord{F} \ccsIdentifier{P_2}$.
  This can be expressed by the HML formula $\hmlObs{a}\hmlNeg\hmlObs{d} \in \obs{F}$,
  which distinguishes $\ccsIdentifier{P_1}$ from $\ccsIdentifier{P_2}$.---Moreover,
  $\ccsIdentifier{P_2}$ cannot simulate the transition $\ccsIdentifier{P_1} \step{a} \ccsInm{b}\ccsChoice\ccsInm{c}$,
  so $\ccsIdentifier{P_1} \not\bPreord{1S} \ccsIdentifier{P_2}$.
  This can be expressed by the distinguishing formula $\hmlObs{a}\hmlAndS\{\hmlObs{b}, \hmlObs{c}\} \in \obs{1S}$.
  (This formula happens to be in $\obs{R}$, the readiness observations, as well.)

  As a more complex example from~\cite[p.~21]{glabbeek2001ltbtsiReport}, see \refFig{fig:ts-example-P3P4}.
  $\ccsIdentifier{P_3} = \ccsInm{a}\ccsPrefix(\ccsInm{b}\ccsChoice\ccsInm{c}\ccsPrefix\ccsInm{d}) \ccsChoice \ccsInm{a}\ccsPrefix(\ccsInm{f}\ccsChoice\ccsInm{c}\ccsPrefix\ccsInm{e})$ can be distinguished from $\ccsIdentifier{P_4} = \ccsInm{a}\ccsPrefix(\ccsInm{b}\ccsChoice\ccsInm{c}\ccsPrefix\ccsInm{e}) \ccsChoice \ccsInm{a}\ccsPrefix(\ccsInm{f}\ccsChoice\ccsInm{c}\ccsPrefix\ccsInm{d})$ in several ways:
  \begin{itemize}
  \item	$\ccsIdentifier{P_4}$ cannot simulate the transition $\ccsIdentifier{P_3} \step{a} \ccsInm{b}\ccsChoice\ccsInm{c}\ccsPrefix\ccsInm{d}$.
  This means that $\ccsIdentifier{P_3} \not\bPreord{1S} \ccsIdentifier{P_4}$.
  This corresponds to the HML formula $\hmlObs{a}\hmlAndS\{\hmlObs{b}, \hmlObs{c}\hmlObs{d}\} \in \obs{1S}$.
  \item	$\ccsIdentifier{P_3}$ has the failure trace $\textcolor{red}{\varnothing}\ccsInm{a}\textcolor{red}{\{\ccsInm{f}\}}\ccsInm{c}\textcolor{red}{\varnothing}\ccsInm{d}\textcolor{red}{\varnothing}$,
  which is not a failure trace of $\ccsIdentifier{P_4}$,
  so $\ccsIdentifier{P_3} \not\bPreord{FT} \ccsIdentifier{P_4}$.
  This corresponds to the HML formula $\hmlObs{a}\hmlAndS\{\hmlNeg\hmlObs{f}, \hmlObs{c}\hmlObs{d} \} \in \obs{FT}$.
  \item	$\langle \ccsInm{a}, \textcolor{red}{\{\ccsInm{b}, \ccsInm{c}\ccsPrefix\ccsInm{d}\}}\rangle$ is an impossible future of $\ccsIdentifier{P_3}$,
  meaning that after action sequence $\ccsInm{a}$ it may reach a state where no trace from $\{\ccsInm{b}, \ccsInm{c}\ccsPrefix\ccsInm{d}\}$ can be executed,
  but this is not an impossible future of $\ccsIdentifier{P_4}$,
  so $\ccsIdentifier{P_3} \not\bPreord{IF} \ccsIdentifier{P_4}$.
  This corresponds to the HML formula $\hmlObs{a}\hmlAndS\{\hmlNeg\hmlObs{b}, \hmlNeg\hmlObs{c}\hmlObs{d}\} \in \obs{IF}$.
  \end{itemize}
  The three formulas mentioned (and a variant of each one, using the other branch of $\ccsIdentifier{P_3}$) distinguish $\ccsIdentifier{P_3}$ from $\ccsIdentifier{P_4}$.
  Note that none of the three languages are contained in another;
  they are all minimal ways to tell the processes apart.
\end{exa}

\begin{rem}
  Like Ku{\v{c}}era and Esparza~\cite{ke1999logicalProcessQuotients}, who studied the properties of ``good'' observation languages, we glimpse over completed trace, completed simulation and possible worlds observations in \refDef{def:ltbts}, because these observations need a special exhaustive $\hmlAnd{a}{\actions}\varphi_a$, where the $\varphi_a$ are deactivated actions for completed traces, and more complex trees for possible worlds.
  While it could be provided for with additional operators, it would break the closure property of observation languages, without giving much in return.
  For instance, for $\actions = \{a,b\}$, completed trace and completed simulation observations contain the observation $\hmlAndS\{ \hmlNeg\hmlObs{a}, \hmlNeg\hmlObs{b} \}$, but not its subformula $\hmlNeg\hmlObs{a}$.
\end{rem}

\subsection{\label{subsec:priced-formulas}Pricing Formulas}

In our following quest for the coarsest behavioral preorders distinguishing two states,
we actually are interested in the formulas that are part of the \emph{minimal observation languages} from the spectrum (\refDef{def:ltbts}).
We can think of the amount of $\hml$-expressiveness used by a formula as its \emph{price.}
So, our first contribution is overlaying the spectrum with a price metric.

\begin{defi}[Formula price lattice]\label{def:price-lattice}
The formula price lattice $\prices$ is the (complete) lattice over $(\mathbb{N} \cup \{ \infty \})^6$ with the partial order $\sqsubseteq$ defined by pointwise comparison,
that is, $e \sqsubseteq e'$ iff $e_j \leq e'_j$ and $\bigsqcup E = e'$ iff $e'_j = \sup_{e \in E} e_j$,  for all $j = 1, \ldots, 6$.
\end{defi}

We use the six dimensions to catch the price structure of the spectrum from \refDef{def:ltbts}.
Intuitively, we employ the following metrics (listed in the order of the dimensions):

\begin{description}
  \item[1.\ Observations] How many observations $\hmlObs{a} \dots$ may one pass at most when descending down the syntax tree? (So we count levels of observations, not the total number of observations.)
  This is called the ``depth'' of modal operator nesting for a formula in~\cite{hennessy1985hml,milner1990}.
  \item[2.\ Conjunctions] How often may one run into a conjunction? Negations in the beginning or following an observation are counted as implicit conjunctions.
  \item[3.\ Positive Deep Branches] How many positive deep branches may appear in each conjunction?
  We call conjuncts of the form $\hmlObs{a}$ \emph{positive flat branches,}
  and other positive branches \emph{positive deep branches.}
  \item[4.\ Positive Branches] How many positive branches may appear in each conjunction, regardless of their depth?
  \item[5.\ Negations] How many negations may be visited when descending?
  In other variants of HML, this is sometimes called the number of ``alternations'' between $\Diamond$ and $\Box$.
  \item[6.\ Negated Observations] How many observations can happen under each negation?
  In other words, what is the maximum observation depth of the negative parts of the formula?
\end{description}
Each dimension expresses some aspect of the complexity of formulas in the spectrum of \refDef{def:ltbts}.
Intuitively, the notions of equivalence in the spectrum differ by how linear observations (traces) and branching observations (conjunctions) can be mixed.
At the branching points, they differ in the kinds of lower bounds (positive branches) and upper bounds (negative branches) that may be imposed on follow-up behavior.
This also gives an intuitive explanation of why negations under observations count as implicit conjunctions in our metric:
they form an upper bound on the behavior at the point where the trace observation stops.

More formally, we compute the six dimensions as follows:

\begin{defi}[Formula expressiveness prices]\label{def:formula-expressiveness}
  We define the \emph{(expressiveness) price} of a formula, $\expr \colon \hml \rightarrow \prices$ recursively by:
\begin{align*}
  \expr(\hmlObs{a}\varphi) &=
  \begin{pmatrix}
    1 + \expr_1(\varphi)\\
    0\\
    0\\
    0\\
    0\\
    0
  \end{pmatrix} \sqcup \expr(\widehat{\varphi})
\qquad\qquad
  \expr(\hmlNeg\varphi) =
  \begin{pmatrix}
    0\\
    0\\
    0\\
    0\\
    1 + \expr_5(\varphi)\\
    \expr_1(\varphi)
  \end{pmatrix} \sqcup \expr(\varphi)
\\[0.5\baselineskip]
  \textstyle\expr(\bigwedge\limits_{i \in I}\varphi_i) &=
  \begin{pmatrix}
    0\\
    \sup_{i \in I} 1 + \expr_2(\varphi_i)\\
    \lvert \mathit{pb} \rvert - \lvert \mathit{pf} \rvert\\
    \lvert \mathit{pb} \rvert\\
    0\\
    0
  \end{pmatrix} \sqcup {\textstyle \bigsqcup\limits_{i \in I}\expr(\varphi_i)}
\qquad\quad
  \begin{matrix*}[l]
    \widehat{\varphi} =
      \begin{cases}
        \hmlAndS \{\varphi\}, & \text{if } \exists \varphi' \ldotp \varphi = \neg \varphi'\\
        \varphi, & \text{otherwise}
      \end{cases}\\ \\
    \mathit{pb} =\{\varphi_i \mid \nexists \varphi' \ldotp \varphi_i = \neg\varphi'\}\\
    \mathit{pf} =\{\varphi_i \mid \exists a \in \actions \ldotp \varphi_i = \hmlObs{a}\}
  \end{matrix*}
\end{align*}
In these calculations, $\expr_j(\,\cdot\,)$ stands for the $j$th dimension of the expressiveness price,
$\mathit{pb}$ is the set of positive branches of a conjunction,
and $\mathit{pf}$ is the set of positive flat branches.

The price of a standalone formula $\varphi$ is $\expr(\widehat{\varphi})$.
(We also apply $\widehat{\,\cdot\,}$ in the definition of $\expr(\hmlObs{a}\varphi)$ to ensure
that negations following an observation are counted as implicit conjunctions.)
\end{defi}

\begin{exa}%
  \label{exa:price-formulas}
  Let us calculate the prices of the formulas in \refExample{exa:distinguishing-formulas}.
  For $\hmlObs{a}\hmlNeg\hmlObs{d}$ (that distinguishes $\ccsIdentifier{P_1}$ from $\ccsIdentifier{P_2}$),
  the price of subformula $\hmlNeg\hmlObs{d}$ is calculated with an additional conjunction:
  $\expr(\widehat{\hmlNeg\hmlObs{d}}) = \expr(\hmlAndS\{\hmlNeg\hmlObs{d}\}) = (1,1,0,0,1,1)$.
  This leads to $\expr(\stretchedhat{1.5}{\hmlObs{a}\hmlNeg\hmlObs{d}}) = (2,1,0,0,1,1)$.%
  ---%
  The other distinguishing formula $\hmlObs{a}\hmlAndS\{\hmlObs{b}, \hmlObs{c}\}$ has price $(2,1,0,2,0,0)$.

  The formulas that distinguish $\ccsIdentifier{P_3}$ from $\ccsIdentifier{P_4}$ have the prices:
  \begin{gather*}
  \expr(\stretchedhat{3.7}{\hmlObs{a}\hmlAndS\{\hmlObs{b}, \hmlObs{c}\hmlObs{d}\}}) = (3,1,1,2,0,0)
  \qquad\quad
  \expr(\stretchedhat{4.05}{\hmlObs{a}\hmlAndS\{\hmlNeg\hmlObs{f}, \hmlObs{c}\hmlObs{d}\}}) = (3,1,1,1,1,1) \\
  \expr(\stretchedhat{4.4}{\hmlObs{a}\hmlAndS\{\hmlNeg\hmlObs{b}, \hmlNeg\hmlObs{c}\hmlObs{d}\}}) = (3,1,0,0,1,2)
  \end{gather*}
  For both process pairs, there are multiple minimal-price distinguishing formulas,
  which are incomparable.
  This reflects our earlier observation that there are multiple minimal languages to tell the processes apart.
  We will make this more exact in \refLem{lem:obslang-price-characterization} below.
\end{exa}

We say that a formula $\varphi_1$ \emph{dominates} $\varphi_2$ if $\varphi_1$ has lower or equal values than $\varphi_2$ in each dimension of the metrics with at least one entry strictly lower, for which we write $\expr(\widehat{\varphi_1}) \sqsubset \expr(\widehat{\varphi_2})$.

\begin{table}[tb]
  \renewcommand*\arraystretch{1.25}
  \centering
  \caption{Expressiveness price bounds per observation language.}%
  \label{tab:ltbts-dimensions}
  \begin{tabular}{@{}lccccccc@{}}
  \toprule
  Observation language          & \rotatebox{75}{Observations} & \rotatebox{75}{Conjunctions} & \rotatebox{75}{Positive deep br.} & \rotatebox{75}{Positive branches} & \rotatebox{75}{Negations} & \rotatebox{75}{Negated observations} \\ \midrule
  enabledness $\obs{E}$         & 1        & 0            & 0             & 0             & 0         & 0               \\
  trace $\obs{T}$               & $\infty$ & 0            & 0             & 0             & 0         & 0               \\
  failure $\obs{F}$             & $\infty$ & 1            & 0             & 0             & 1         & 1               \\
  readiness $\obs{R}$           & $\infty$ & 1            & 0             & $\infty$      & 1         & 1               \\
  failure-trace $\obs{FT}$      & $\infty$ & $\infty$     & 1             & 1             & 1         & 1               \\
  ready-trace $\obs{RT}$        & $\infty$ & $\infty$     & 1             & $\infty$      & 1         & 1               \\
  impossible-future $\obs{IF}$  & $\infty$ & 1            & 0             & 0             & 1         & $\infty$        \\
  possible-future $\obs{PF}$    & $\infty$ & 1            & $\infty$      & $\infty$      & 1         & $\infty$        \\
  ready-simulation $\obs{RS}$   & $\infty$ & $\infty$     & $\infty$      & $\infty$      & 1         & 1               \\
  $(n{+}1)$-nested-simulation $\obs{(\mathit{n}+\mathrm{1})S}$%
                                & $\infty$ & $\infty$     & $\infty$      & $\infty$      & $n$       & $\infty$        \\
  bisimulation $\obs{B}$        & $\infty$ & $\infty$     & $\infty$      & $\infty$      & $\infty$  & $\infty$        \\
  \bottomrule
  \end{tabular}
\end{table}

\begin{exa}\label{exa:locally-more-expressive}
  A locally more expensive formula may pay off as part of a bigger global formula. For example, $\hmlNeg\hmlObs{a}$ is dominated by $\hmlObsI{b}$ as (1,1,0,0,1,1) $\sqsupset$ (1,0,0,0,0,0). But inserted into the context $\hmlAndS \{ \hmlNeg\hmlObs{c}, ~\cdot~ \}$, the more expensive negation $\hmlNeg\hmlObs{a}$ leads to price (1,1,0,0,1,1) as opposed to the price (1,1,0,1,1,1) after inserting the positive observation $\hmlObsI{b}$.

Note that this inversion of domination can only happen
for the dimensions $\expr_3$ (positive deep branches) and $\expr_4$ (positive branches),
as all other dimensions in \refDef{def:formula-expressiveness} are clearly monotonic.
Then the context must contain a conjunction $\hmlAnd{i}{I}\varphi_i$
and the local formula must be a conjunct $\varphi_i$,
the cheaper local formula is an observation formula and the other a negation formula.
\end{exa}

Table~\ref{tab:ltbts-dimensions} gives an overview of how many syntactic $\hml$-features the observation languages of the spectrum (\refDef{def:ltbts}) may use at most%
---%
these are the least upper bounds of the prices for the contained observations.
So, we are talking \emph{budgets,} in the price analogy.

\begin{lem}\label{lem:obslang-price-characterization}
  A formula $\varphi$ is in an observation language $\obs{\mathit X}$ with expressiveness price bound $e_X$ from Table~\ref{tab:ltbts-dimensions}
  precisely if its price is within the bound,
  that is $\expr(\widehat{\varphi})\sqsubseteq e_X$.
\end{lem}
\begin{proof}[Proof Sketch]
  $\varphi\in \obs{\mathit X}$ implies $\expr(\widehat{\varphi})\sqsubseteq e_X$ as the $e_X$ in the table
  exactly are the least upper bounds $\bigsqcup_{\varphi \in \obs{\mathit X}} \expr(\widehat{\varphi})$.
  That $\expr(\widehat{\varphi})\sqsubseteq e_X$ implies $\varphi\in \obs{\mathit X}$, is more involved.
  Basically it amounts to constructing the fiber function $\expr^{-1}$
  and noticing that a partial evaluation of $\expr^{-1}(e_X)$ looks exactly like the cases in \refDef{def:ltbts}.
  For failures for instance, one would get $\expr^{-1}(e_\mathrm{F}) = \{\hmlObs{a}\varphi \mid \varphi \in \expr^{-1}(e_\mathrm{F})\} \cup \{\hmlTrue\} \cup \{ \hmlAnd{a}{A} \hmlNeg\hmlObs{a} \mid A \subseteq \actions \}$.
\end{proof}

Note that no strict subset of these dimensions distinguishes all languages of \refDef{def:ltbts} and $\obs{E}$.

\begin{exa}\label{exa:prices-distinguishing-formulas}
  We can now compare the formula prices calculated in \refExample{exa:price-formulas} with the expressiveness price bounds:
  The price of $\hmlObs{a}\hmlNeg\hmlObs{d} \in \obs{F}$ is $(2,1,0,0,1,1)$,
  which is cheaper than the price bound
  $e_\mathrm{F} = (\infty,1,0,0,1,1)$ of $\obs{F}$.---The
  other distinguishing formula $\hmlObs{a}\hmlAndS\{\hmlObs{b}, \hmlObs{c}\} \in \obs{1S} \cap \obs{R}$ has price $(2,1,0,2,0,0)$,
  which is cheaper than the price bounds of $\obs{1S}$ and $\obs{R}$.

  The formulas that distinguish $\ccsIdentifier{P_3}$ from $\ccsIdentifier{P_4}$ are within the following price bounds:
  \[\begin{array}{@{}r@{}c@{\,}c@{\,}c@{\,}c@{\,}c@{\,}c@{}l@{}}
  \expr(\stretchedhat{3.7}{\hmlObs{a}\hmlAndS\{\hmlObs{b}, \hmlObs{c}\hmlObs{d}\}}) = (3,1,1,2,0,0)
    & {} \sqsubseteq (\infty, & \infty, & \infty, & \infty, & 0, & \infty & )\text{, the price bound of }\obs{1S} \\
    & \textcolor{gray}{{} \sqsubseteq (\infty,} & \textcolor{gray}{1,} & \textcolor{gray}{\infty,} & \textcolor{gray}{\infty,} & \textcolor{gray}{1,} & \textcolor{gray}{\infty} & \textcolor{gray}{)\text{, the price bound of }\obs{PF}} \\
    & \textcolor{gray}{{} \sqsubseteq (\infty,} & \textcolor{gray}{\infty,} & \textcolor{gray}{1,} & \textcolor{gray}{\infty,} & \textcolor{gray}{1,} & \textcolor{gray}{1} & \textcolor{gray}{)\text{, the price bound of }\obs{RT}} \\
  \expr(\stretchedhat{4.05}{\hmlObs{a}\hmlAndS\{\hmlNeg\hmlObs{f}, \hmlObs{c}\hmlObs{d}\}}) = (3,1,1,1,1,1) & {} \sqsubseteq (\infty, & \infty, & 1, & 1, & 1, & 1 & )\text{, the price bound of }\obs{FT} \\
    & \textcolor{gray}{{} \sqsubseteq (\infty,} & \textcolor{gray}{1,} & \textcolor{gray}{\infty,} & \textcolor{gray}{\infty,} & \textcolor{gray}{1,} & \textcolor{gray}{\infty} & \textcolor{gray}{)\text{, the price bound of }\obs{PF}} \\
  \expr(\stretchedhat{4.4}{\hmlObs{a}\hmlAndS\{\hmlNeg\hmlObs{b}, \hmlNeg\hmlObs{c}\hmlObs{d}\}}) = (3,1,0,0,1,2) & {} \sqsubseteq (\infty, & 1, & 0, & 0, & 1, & \infty & )\text{, the price bound of }\obs{IF}
  \end{array}\]
  We can see that these prices are below some further price bounds printed in gray.
  However, these bounds are less relevant,
  as $\obs{FT} \subset \obs{RT}$ and $\obs{IF} \subset \obs{PF}$.
\end{exa}

Not every conceivable observation language can soundly be characterized by our metric.
For instance, the ``two-$\action a$s-may-happen-equivalence'' with observation language $\{\hmlObs{a}\hmlObs{a}, \hmlTrue\}$
would have coordinates $(2,0,0,0,0,0)$
but does not contain the formula $\hmlObs{a}$
even though $\expr(\hmlObs{a})=(1,0,0,0,0,0)$ is below.
However, all common closed observation languages have characteristic coordinates in our price lattice.


\begin{rem}[New prices]%
  \label{rem:new-prices}
  \refDef{def:formula-expressiveness} of formula expressiveness prices and the spectrum characterization differ from the conference version~\cite{bisping2021ltbtsTacas}
  in that the fourth dimension counts positive branches instead of positive \emph{flat} branches
  and in that the last dimension ``negated observations'' measures the observation depth instead of negated syntax tree height.
  We will use these two changes in the proof of \refThm{thm:correctness-general}.

  Not counting flat positive branches resolves some technical problems that existed around readiness and failure trace languages.
  The advantage of this definition of formula price is
  that we have $\expr(\hmlAndS\{\hmlObs{a}\hmlObs{a}, \hmlObs{b}\}) \sqsubseteq \expr(\hmlAndS\{\hmlObs{a}\hmlObs{a}, \hmlObs{b}\hmlObs{b}\})$.
  This would not hold in the definition of~\cite{bisping2021ltbtsTacas},
  because $\hmlAndS\{\hmlObs{a}\hmlObs{a}, \hmlObs{b}\}$ contains the positive flat branch $\hmlObs{b}$
  while the other formula does not contain any positive flat branches.

  The measuring of negated height in terms of observations is desirable for technical reasons.
  It is more stable under the insertion of virtual conjunctions in front of negations.
  This way, $\hmlNeg\hmlObs{a}\hmlAndS\{ \hmlNeg\hmlObs{b}, \hmlNeg\hmlObs{a}\}$,
   $\hmlNeg\hmlObs{a}\hmlAndS \{ \hmlNeg\hmlObs{b} \}$ and
   $\hmlNeg\hmlObs{a}\hmlNeg\hmlObs{b}$
   all have $\expr(\widehat{\,\cdot\,}) = (2, 2, 0, 0, 2, 2)$ as price.
   Using syntax tree height, the sixth dimension would be $3$ for the last formula and $4$ for the other two.
   Also, if we would not include an implicit conjunction before each negation,
   the last formula would have a smaller value in the second dimension.
\end{rem}

\subsection{Double Negations and Negated Conjunctions}%
  \label{subsec:negated-conjunctions}
  When searching distinguishing formulas in the observation languages considered,
  double negations and negated conjunctions are not useful and can be suppressed.

  A \emph{double negation} $\hmlNeg\hmlNeg\varphi$ is more expensive than $\varphi$ in most contexts.
  Only in a conjunction, a double negation can mask a positive branch;
  for example, $\hmlObs{a}\hmlAndS\{\hmlNeg\hmlNeg\hmlObs{b}, \linebreak[0] \hmlNeg\hmlObs{c}\hmlNeg\hmlObs{d}\}$
  is strictly cheaper than $\hmlObs{a}\hmlAndS\{\hmlObs{b}, \linebreak[0] \hmlNeg\hmlObs{c}\hmlNeg\hmlObs{d}\}$.
  Note that a formula containing a double negation is in 3-nested simulation $\obs{3S}$ or a larger language,
  where the number of positive (deep) branches is irrelevant.
  Therefore, double negations are not useful,
  and we forbid them in distinguishing formulas.

  \emph{Negated conjunctions} are disjunctions.
  (In this paragraph, we use the abbreviation $\bigvee_{i \in I}\varphi_i$ for $\hmlNeg\hmlAnd{i}{I}\hmlNeg\varphi_i$.)
  They can be replaced by (almost always cheaper) formulas without disjunction.
  For example, $\hmlObs{a}\hmlNeg\hmlAndS\{\hmlObs{b},\hmlObs{c}\}$ distinguishes $\ccsInm{a}\ccsPrefix\ccsInm{b}$ and $\ccsInm{a}\ccsPrefix\ccsInm{c}$ from $\ccsInm{a}\ccsPrefix(\ccsInm{b}\ccsChoice\ccsInm{c})$.
  This formula is equivalent to $\bigvee\{\hmlObs{a}\hmlNeg\hmlObs{b},\hmlObs{a}\hmlNeg\hmlObs{c}\}$.
  Note that the disjuncts $\hmlObs{a}\hmlNeg\hmlObs{b}$ and $\hmlObs{a}\hmlNeg\hmlObs{c}$
  also distinguish $\ccsInm{a}\ccsPrefix\ccsInm{c}$ and $\ccsInm{a}\ccsPrefix\ccsInm{b}$, respectively, from $\ccsInm{a}\ccsPrefix(\ccsInm{b}\ccsChoice\ccsInm{c})$.

  In general, any formula containing a disjunction can be transformed into a kind of disjunctive normal form $\bigvee_{i \in I}\varphi_i$,
  where the $\varphi_i$ do not contain any disjunctions.
  When a concrete process $p$ needs to be distinguished from some process $q$,
  one of the disjuncts $\varphi_i$ can be used instead of the whole formula.
  The transformation to disjunctive normal form will almost always lead to disjuncts that are cheaper than or as cheap as the original formula.
  The only exception is the case where a disjunction masks a double negation in a conjunction;
  for example, $\hmlObs{a}\hmlAndS\{\hmlNeg\hmlAndS\{\hmlNeg\hmlObs{b}\}, \hmlNeg\hmlObs{c}\hmlNeg\hmlObs{d}\}$
  corresponds to $\hmlObs{a}\hmlAndS\{\hmlNeg\hmlNeg\hmlObs{b}, \linebreak[0] \hmlNeg\hmlObs{c}\hmlNeg\hmlObs{d}\}$.
  As explained above, this only happens in 3-nested simulation $\obs{3S}$ and larger languages,
  where the number of positive (deep) branches is irrelevant.
  Therefore, like double negations, negated conjunctions are not useful,
  and we forbid them in distinguishing formulas as well.

\section{A Game to Find Distinguishing Formulas\label{sec:distinguishing-formulas-game}}

This section introduces our main contribution: the spectroscopy game (\refDef{def:spectroscopy-game}), and how to build all interesting distinguishing $\hml$ formulas from its winning region (\refDef{def:strategy-formulas}).

\subsection{The Abstract Observation Preorder Problem\label{subsec:spectroscopy-problem}}

In what follows, we generalize the problem whether an observation language preorders two states (\refDef{def:observation-preorder}) in two ways:

\begin{enumerate}
  \item We not only consider one observation language $\obs{\mathit X}$ from \refDef{def:ltbts} or $\obs{E}$
  but all of them at the same time.
  \item We do not compare one process $p$ to another $q$, but rather one $p$ to a set $Q \subseteq \proc$.
\end{enumerate}
The first generalization is the main objective of this paper. The second one enables the construction of a uniform algorithm. The abstracted problem thus becomes:

\begin{prob}[Abstract observation preorder problem]\label{prob:abstract-obs-problem}
  Given a process $p$ and a set of processes $Q$,
  what are the observation languages from \refDef{def:ltbts} (including $\obs{E}$)
  for which $p$ is preordered to every $q \in Q$?
\end{prob}

Our approach to solve this problem looks for the set of minimal languages to tell the processes apart.
We characterize these minimal distinguishing languages through a set of coordinates from the price lattice (\refDef{def:price-lattice}),
where every coordinate is justified by a distinguishing formula with this price.
In line with \refSubsec{subsec:negated-conjunctions}, we do not care about formulas with double negations and negated conjunctions.

\begin{prob}[Cheapest distinction problem]\label{prob:cheapest-distinction}
  Given a process $p$ and a set of processes $Q$,
  what is the set of minimal prices (according to \refDef{def:formula-expressiveness}) of formulas with neither double negations nor negated conjunctions
  that distinguish $p$ from every $q \in Q$?
  What are illuminating witness formulas for each such price?
\end{prob}

There is a straightforward way of turning the problem whether an observation language $\obs{\mathit X}$ preorders $p$ and $q$ into a game: Have the attacker pick a supposedly distinguishing formula $\varphi \in \obs{\mathit X}$, and then have the defender choose whether to play the $\hml$ game (\refDef{def:hml-game}) for $\hmlSemantics{\hmlNeg\varphi}{}{p}$ or for $\hmlSemantics{\varphi}{}{q}$. One can examine the attacker winning strategies comprised of price-minimal formulas and solve Problem~\ref{prob:cheapest-distinction}. This direct route will yield infinite games for infinite $\obs{\mathit X}$---and all the languages from \refDef{def:ltbts} are infinite.

To bypass the infinity issue, the next subsection will introduce a variation of this game
\emph{where the attacker gradually chooses their attacking formula implicitly.}
In particular, this means that the attacker decides which observations to play.
In return, the defender does not need to pick a side in the beginning
and may postpone the decision where (on the side of $\hmlSemantics{\varphi}{}{q}$)
an observation leads.
Postponing decisions here means that the defender state is modeled non-deterministically, moving to multiple process states at once.
The mechanics are analogous to the standard powerset construction when transforming non-deterministic finite automata into deterministic ones.
In effect, the attacker tries to show that a formula $\varphi \in \obs{\mathit X}$ distinguishes $p$ from every $q \in Q$, and the defender tries to prove that no such formula exists.

\begin{figure}[t]
  \begin{center}
      \begin{tikzpicture}[>->,shorten <=1pt,shorten >=0.5pt,auto,node distance=2cm, rel/.style={dashed,font=\it},
        posStyle/.style={draw, inner sep=1ex,minimum size=1cm,minimum width=2cm,anchor=center,draw,black,fill=gray!5}]
          \node[posStyle]
            (Att){$\attackerPos{p,Q}$};
          \node[posStyle, dashed]
            (AttObs) [right of=Att, node distance=8.5cm] {$\attackerPos{p^\prime,Q^\prime}$};
          \node[posStyle, dashed]
            (AttConj) [above of=AttObs, node distance=4cm] {$\attackerPos{p,\{q\}}$};
          \node[ellipse, draw, inner sep=1ex, minimum size=1cm,minimum width=2cm,fill=gray!5]
            (Def) [left of=AttConj, node distance=5.5cm] {$\defenderPos{p,\partition{Q}}$};
          \node[posStyle]
            (AttConj2) [below right of=Def, node distance=3cm] {$\attackerPos{p,Q^*}^\notand$};
          \node[posStyle, dashed]
            (AttSwap) [below of=AttObs] {$\attackerPos{q,\{p\}}$};

          \path
            (Att) edge[] node[align=center, label={[label distance=0.15cm]below:$\textcolor{gray}{\hmlObs{a}}$}] {$p\step{a}p^\prime \land$\\ $Q'=\{q^\prime \mid \exists q \in Q\ldotp q \step a q^\prime\}$} (AttObs)
            (Att) edge[bend left=20] node[label={[label distance=0.25cm]-10:$\textcolor{gray}{\land}$}] {$\partition{Q}$ partitions $Q$} (Def)
            (Def) edge[] node[label=below:$\textcolor{gray}{*}$] {$\{q\} \in \partition{Q}$} (AttConj)
            (Def) edge[bend right=20] node[align=center, pos=0.7, label={[label distance=0.3cm]-170:$\textcolor{gray}{*}$}] {$Q^* \in \partition{Q}$\\ not a singleton} (AttConj2)
            (AttConj2) edge[bend left=20] node[align=center, label=-170:$\textcolor{gray}{\hmlObs{a}}$] {$p\step{a}p^\prime \land$\\ $Q'=\{q^\prime \mid \exists q \in Q^*\ldotp q \step a q^\prime\}$} (AttObs)
            (Att) edge[bend right=15] node[label={[label distance=0.3cm]-100:$\textcolor{gray}{\hmlNeg}$}] {$Q = \{q\}$} (AttSwap);

      \end{tikzpicture}
  \end{center}
  \caption{
    Schematic spectroscopy game $\gameSpectroscopy$ of \refDef{def:spectroscopy-game}.}%
  \label{fig:spec-game}
\end{figure}

\subsection{The Spectroscopy Game\label{subsec:spectroscopy-game}}

Let us now look at the ``spectroscopy game.'' It forms the heart of this paper. Figure~\ref{fig:spec-game} gives a graphical representation.

\begin{defi}[Spectroscopy game]%
  \label{def:spectroscopy-game}
  For a transition system $\system=(\proc,\actions,\step{})$,
  the $L$-labeled \emph{spectroscopy game} $\gameSpectroscopy^\system[g_0]=(G,G_\defenderSubscript,\gameMoveX{\cdot},g_{0})$
  with $L = \{\neg, \land, *, \hmlObs{a} \mid a \in \actions \} $
  consists of
  \begin{itemize}
      \item \emph{attacker positions} $\attackerPos{p,Q} \in G_\attackerSubscript$ with $p \in \proc$, $Q \in \powerSet{\proc}$,
      \item \emph{post-conjunction attacker positions} $\attackerPos{p,Q}^\notand \in G_\attackerSubscript$,
      where $Q \in \powerSet{\proc}$ contains at least two elements; the symbol ``$\notand$'' indicates that conjunct challenge moves are not allowed there,
        \item \emph{defender positions} $\defenderPos{p,\partition{Q}} \in G_\defenderSubscript$ where $\partition{Q} \subseteq \powerSet{\proc}$ is a partition of a subset of $\proc$,
  \end{itemize}
  and four kinds of moves:

  \medskip

  \noindent
  \begin{tabularx}{\linewidth}{@{}l@{\;}l@{\;}c@{\;}l@{\;}X@{}}
    \textbullet\hspace{\labelsep}\emph{observation moves}
    & $\attackerPos{p,Q}$ & $\gameMoveX{\hmlObs{a}}$ & $\attackerPos{p^\prime,\{q^\prime \mid \exists q \in Q\ldotp q \step{\ccsInm{a}} q^\prime\}}$ & if $p \step{\ccsInm{a}} p^\prime$,\\
    & $\attackerPos{p,Q}^\notand$ & $\gameMoveX{\hmlObs{a}}$ & $\attackerPos{p^\prime,\{q^\prime \mid \exists q \in Q\ldotp q \step{\ccsInm{a}} q^\prime\}}$ & if $p \step{\ccsInm{a}} p^\prime$,\\
    \textbullet\hspace{\labelsep}\emph{conjunct challenges}
    & $\attackerPos{p,Q}$ & $\gameMoveX{\land}$ & $\defenderPos{p,\partition{Q}}$ &
    if $\partition{Q}$ is a partition of $Q$ that is not trivial (\ie $\partition{Q} \not= \{ Q \}$),
    \\
    \textbullet\hspace{\labelsep}\emph{conjunct answers}
    & $\defenderPos{p,\partition{Q}}$ & $\gameMoveX{*}$ & $\attackerPos{p,Q}^\notand$ & if $Q \in \partition{Q}$ is not a singleton,
    \\
    & $\defenderPos{p,\partition{Q}}$ & $\gameMoveX{*}$ & $\attackerPos{p,\{q\}}$ & if $\{q\} \in \partition{Q}$, and
    \\
    \textbullet\hspace{\labelsep}\emph{negation moves}
    & $\attackerPos{p,\{q\}}$ & $\gameMoveX{\hmlNeg}$ & $\attackerPos{q,\{p\}}$.
  \end{tabularx}
\end{defi}

\noindent
Attacker moves are labeled with the syntactic $\hml$ constructs from which they originate. This does not change expressive power but is helpful for formula reconstruction in the next section.
Accordingly, attacker strategies for spectroscopy games are subsets of labeled moves.


\begin{exa}%
  \label{exa:spec-game-1}
  Let us say we want to compare the processes $\ccsInm{a}\ccsPrefix\ccsInm{b}$ and $\ccsInm{a}\ccsPrefix(\ccsInm{a}\ccsChoice \ccsInm{b}) \ccsChoice \ccsInm{a}\ccsPrefix\ccsInm{b}\ccsPrefix\ccsInm{b} \ccsChoice \ccsInm{a}$.
  \refFig{fig:spectro-game-1} shows the game graph.
  Recall from \refDef{def:plays-wins} that the attacker wins in defender positions
  that are stuck, concretely in position $\defenderPos{\ccsStop,\varnothing}$.
  The defender wins infinite plays, concretely from position $\attackerPos{\ccsStop, \{ \ccsStop \}}$,
  where the only play loops infinitely in place.
  Elsewhere, the player that can force the play into one of these two positions wins.
  The black part represents the attacker winning region,
  which corresponds to processes that can be distinguished by HML formulas.
  The edge labels already hint at how to construct such formulas.
  The (thin, small) red part corresponds to processes that cannot be distinguished and thus are won by the defender.
  We will make this more precise in the next subsections.
\end{exa}

\begin{figure}[tb]
  \begin{center}
  \begin{tikzpicture}[shorten <=1pt,shorten >=0.5pt]

    \path[use as bounding box,draw=none] (7.56,-4.3) rectangle (-7.34,9.37);


    \begin{scope}[every node/.style={inner sep=1ex,anchor=center,draw,black,fill=gray!5}]

    \node(s0) at(-0.5,9) {$\attackerPos{\ccsInm{a}\ccsPrefix\ccsInm{b}, \{ \ccsInm{a}\ccsPrefix(\ccsInm{a}\ccsChoice \ccsInm{b}) \ccsChoice \ccsInm{a}\ccsPrefix\ccsInm{b}\ccsPrefix\ccsInm{b} \ccsChoice \ccsInm{a} \}}$};

    \node(s0neg) at(3.5,7.75) {$\attackerPos{\ccsInm{a}\ccsPrefix(\ccsInm{a}\ccsChoice \ccsInm{b}) \ccsChoice \ccsInm{a}\ccsPrefix\ccsInm{b}\ccsPrefix\ccsInm{b} \ccsChoice \ccsInm{a}, \{ \ccsInm{a}\ccsPrefix\ccsInm{b} \}}$};

    \node(s1) at(-0.5,6.5) {$\attackerPos{\ccsInm{b}, \{ \ccsInm{a}\ccsChoice \ccsInm{b}, \ccsInm{b}\ccsPrefix\ccsInm{b}, \ccsStop \}}$};

    \node[ellipse,minimum width=12em](s2a) at(-3.1,4.75) {\makebox[0pt]{$\defenderPos{\ccsInm{b}, \{ \{\ccsInm{a}\ccsChoice \ccsInm{b}\}, \{\ccsInm{b}\ccsPrefix\ccsInm{b}, \ccsStop\} \}}$}};

    \node[ellipse,minimum width=13em](s2b) at(2.1,4.75) {\makebox[0pt]{$\defenderPos{\ccsInm{b}, \{ \{\ccsInm{a}\ccsChoice \ccsInm{b}\}, \{\ccsInm{b}\ccsPrefix\ccsInm{b}\}, \{\ccsStop\} \}}$}};

    \node(s3b) at(-1.75,2) {$\attackerPos{\ccsInm{b}, \{ \ccsInm{b}\ccsPrefix\ccsInm{b}, \ccsStop \}}^\notand$};

    \node(s3c) at(1.75,0.5) {$\attackerPos{\ccsInm{b}, \{ \ccsStop \}}$};

    \node(s3d) at(-1.5,0.5) {$\attackerPos{\ccsInm{b}, \{ \ccsInm{b}\ccsPrefix\ccsInm{b} \}}$};

    \node(s3a) at(-4.75,0.5) {$\attackerPos{\ccsInm{b}, \{ \ccsInm{a}\ccsChoice \ccsInm{b} \}}$};

    \node(s4c) at(1.75,-1) {$\attackerPos{\ccsStop, \{ \ccsInm{b} \}}$};

    \node(s4d) at(-1.5,-1) {$\attackerPos{\ccsInm{b}\ccsPrefix\ccsInm{b}, \{ \ccsInm{b} \}}$};

    \node(s4a) at(-4.75,-1) {$\attackerPos{\ccsInm{a}\ccsChoice \ccsInm{b}, \{ \ccsInm{b} \}}$};

    \node(s5) at(-1.5,-3) {$\attackerPos{\ccsStop, \varnothing}$};

    \node[ellipse,inner sep=2pt,minimum width=5.6em](bott) at(1.75,-3) {\makebox[0pt]{$\defenderPos{\ccsStop, \varnothing}$}};

    \end{scope}


    \begin{scope}[every node/.style={line width=0.25pt,inner sep=1ex,anchor=center,draw,darkred!60,fill=darkred!3,node font=\scriptsize}]

    \node[ellipse,minimum width=12.3em](s2a0) at(4.5,3.5) {\makebox[0pt]{$\defenderPos{\ccsInm{b}, \{ \{\ccsInm{a}\ccsChoice \ccsInm{b}, \ccsInm{b}\ccsPrefix\ccsInm{b}\}, \{\ccsStop\} \}}$}};

    \node[ellipse,minimum width=12.3em](s2b0) at(-5.5,3.5) {\makebox[0pt]{$\defenderPos{\ccsInm{b}, \{ \{\ccsInm{a}\ccsChoice \ccsInm{b}, \ccsStop\}, \{\ccsInm{b}\ccsPrefix\ccsInm{b}\} \}}$}};

    \node(s3a0) at(5,2) {$\attackerPos{\ccsInm{b}, \{ \ccsInm{a}\ccsChoice \ccsInm{b}, \ccsInm{b}\ccsPrefix\ccsInm{b} \}}^\notand$};

    \node(s3d0) at(-6,2) {$\attackerPos{\ccsInm{b}, \{ \ccsInm{a}\ccsChoice \ccsInm{b}, \ccsStop \}}^\notand$};

    \node(s50b) at(5,0.5) {$\attackerPos{\ccsStop, \{ \ccsStop, \ccsInm{b} \}}$};

    \node[ellipse,minimum width=8.5em](s50bdef) at(5,-1) {\makebox[0pt]{$\defenderPos{\ccsStop, \{ \{\ccsStop\}, \{\ccsInm{b}\} \}}$}};

    \node(s50) at(-4.75,-3) {$\attackerPos{\ccsStop, \{ \ccsStop \}}$};

    \end{scope}


    \begin{scope}[>->,line width=0.25pt,darkred!60,every node/.style={node font=\scriptsize}]

    \draw (s1.348) .. controls (4.7,5.5) and (5.45,5) .. node[above,pos=.35]{$~\land$} (s2a0);
    \draw (s1) .. controls (-5.5,5.5) and (-6.25,5) .. node[above,pos=.35]{$\land~$} (s2b0);
    \draw (s1.357) .. controls (7.25,6) and (7.625,2) .. node[above,pos=.25]{$\hmlObs{b}$} (s50b);
    \draw (s2a0) -- node[right]{$*$} (s3a0);
    \draw (s2b0) -- node[left]{$*$} (s3d0);
    \draw (s3a) .. controls (-6.5,-0.5) and (-6.5,-1.5) .. node[left,pos=.25]{$\hmlObs{b}$} (s50);
    \draw (s3a0) -- node[left]{$\hmlObs{b}$} (s50b);
    \draw (s3d0) .. controls (-7,0.5) and (-6.75,-1.75) .. node[left,pos=.465]{$\hmlObs{b}$} (s50);
    \draw (s4a) -- node[left]{$\hmlObs{b}$} (s50);
    \draw (s50) to[loop below] node[below]{$\hmlNeg$} ();
    \draw (s50b) -- node[left]{$\land$} (s50bdef);
    \draw (s50bdef) .. controls (3.5,-4.75) and (-0.5,-4) .. node[left,pos=.18]{$*$} (s50);

    \end{scope}


    \begin{scope}[>->,line width=0.25pt,black!60,every node/.style={node font=\scriptsize}]

    \draw (s2a0) -- node[left,pos=.192]{$*~$} (s3c);
    \draw (s2b0) to[bend right=5] node[right,pos=.16]{$~*$} (s3d); 
    \draw (s50bdef) -- node[below,pos=.4]{$*$} (s4c);

    \end{scope}


    \begin{scope}[>->,black] 

    \draw (s0) -- node[left]{$\hmlObs{a}$} (s1);
    \draw ([xshift=-15mm]s0) -- node[left]{$\hmlNeg~~$} (s0neg);
    \draw ([xshift=15mm]s0neg) -- node[right]{$~~\hmlNeg$} (s0);

    \draw (s0neg.333) .. controls (13.75,-1.25) and (-0.35,-8.55) .. node[right,pos=.48]{$~\hmlObs{a}$} (s4a);
    \draw (s0neg.325) .. controls (9.25,3.4) and (8.65,-5.5) .. node[below,pos=.75]{$\hmlObs{a}$} (s4d);
    \draw (s0neg.310) .. controls (8.5,4.25) and (8.5,-4.125) .. node[below,pos=.94,inner sep=2pt]{$\hmlObs{a}$} (s4c);

    \draw (s1) -- node[right]{$~\land$} (s2b);
    \draw (s1) -- node[left]{$\land~$} (s2a);

    \draw (s2b) .. controls (-3.5,2.375) and (-2.625,3.25) .. node[left,pos=.53]{$*~$} (s3a);
    \draw (s2b) -- node[right,pos=.458]{$*$} (s3c);
    \draw (s2b) .. controls (-.75,1.125) and (0,1.875) .. node[left,pos=.219]{$*$} (s3d);
    \draw (s2a) .. controls (-3.325,3.5) and (-3.5,3.25) .. node[right,pos=.368]{$*$} (s3a);
    \draw (s2a) -- node[right,pos=.4]{$*$} (s3b);

    \draw (s4a.70) -- node[right]{$\hmlNeg$} (s3a.290);
    \draw (s3c.250) -- node[left]{$\hmlNeg$} (s4c.110);
    \draw (s3d.250) -- node[left]{$\hmlNeg$} (s4d.110);
    \draw (s4d.70) -- node[right]{$\hmlNeg$} (s3d.290);
    \draw (s3a.250) -- node[left]{$\hmlNeg$} (s4a.110);
    \draw (s4c.70) -- node[right]{$\hmlNeg$} (s3c.290);

    \draw (s3d) -- node[below,pos=.65]{$\hmlObs{b}~$} (s4c);
    \draw (s4d) -- node[above,pos=.65]{$\hmlObs{b}~$} (s3c);

    \draw (s3b) .. controls (-4.25,-0.875) and (-2.75,-2.875) .. node[left,pos=.13]{$\hmlObs{b}$} (s4c);
    \draw (s4a) -- node[below,pos=.45]{$\hmlObs{a}$} (s5);
    \draw (s3c) .. controls (4.5,-1.9) and (1.5,-2) .. node[right,pos=.08]{$~\hmlObs{b}$} (s5);
    \draw (s5) -- node[below]{$\land$} (bott);

  \end{scope}

  \end{tikzpicture}
  \end{center}
  \caption{\label{fig:spectro-game-1}%
  The spectroscopy game that distinguishes process $\ccsInm{a}\ccsPrefix\ccsInm{b}$ from $\ccsInm{a}\ccsPrefix(\ccsInm{a}\ccsChoice \ccsInm{b}) \ccsChoice \ccsInm{a}\ccsPrefix\ccsInm{b}\ccsPrefix\ccsInm{b} \ccsChoice \ccsInm{a}$.
  Parts where the defender wins are drawn thin, small, and \textcolor{darkred!60}{red.}}
\end{figure}

\subsection{Relationship to Bisimulation}%
\label{subsec:relationship-bisimulation}

Comparing $\gameSpectroscopy$ to the standard bisimulation game from the literature (amending games originating from~\cite{stirling1996modal}
with symmetry moves, see e.g.\@~\cite{bnp2020coupledsim30}),
we can easily transfer attacker strategies from there.
In the standard bisimulation game, the attacker will play $(p,q) \gameMove (a, p', q)$ with $p\step{a}p'$
and the defender has to answer by $(a, p', q) \gameMove (p', q')$ with $q \step{a} q'$.
In the spectroscopy game, the attacker can enforce analogous moves by playing
\[
\attackerPos{p,\{q\}} \gameMoveX{\hmlObs{a}} \attackerPos{p', Q'} \gameMoveX{\wedge} \defenderPos{p', \{\{q^*\} \mid q^* \in Q'\}},
\]
which will make the defender pick $\defenderPos{p', \{\{q^*\} \mid q^* \in Q'\}} \gameMoveX{*} \attackerPos{p', \{q'\}}$.

The opposite direction of transfer is not so easy, as the attacker has more ways of winning in $\gameSpectroscopy$. But this asymmetry is precisely why we have to use the spectroscopy game instead of the standard bisimulation game if we want to learn about, for example, interesting failure-trace attacks.

Indeed, the spectroscopy game does characterize bisimilarity. Bisimilarity denotes the existence of a bisimulation relation $\rel R$ that has the properties that $\rel R \subseteq \rel{R}^{-1}$ (symmetry) and that $(\rel{R}^{-1}){\cdot}(\mathord{\step{a}}) \,\subseteq\, (\mathord{\step{a}}){\cdot}(\rel{R}^{-1})$ for all $a\in \actions$ (simulation).

\begin{lem}\label{lem:spectroscopy-game-bisimulation}
  The spectroscopy game $\gameSpectroscopy[\attackerPos{p_0,\{q_0\}}]$ is won by the defender precisely if $p_0$ and $q_0$ are bisimilar.
\end{lem}
\begin{proof}
  We use that bisimilarity is equivalent to the existence of a bisimulation relation, and that winning is equivalent to the existence of a positional strategy.
  \begin{itemize}
    \item Construct the relation $\rel R = \{ (p,q) \mid \attackerPos{p,\{q\}}\in W_\defenderSubscript\}$.
    This must be a symmetric simulation (and thus a bisimulation).
    \begin{itemize}
      \item Symmetry: Assume $(p,q)\in \rel R$. As the attacker can play from $\attackerPos{p,\{q\}}$ to $\attackerPos{q,\{p\}}$, the two can only be in the defender winning region together. Thus $(q,p)\in \rel R$.
      \item Simulation: Assume $(p,q)\in \rel R$ and $p\step{a}p'$. Then there is an attacker move $\attackerPos{p,\{q\}}$ to $\attackerPos{p',Q'}$ for $Q' = \{q'\mid q\step{a}q'\}$.
        We must have $\attackerPos{p',Q'} \in W_\defenderSubscript$
        (otherwise, $\attackerPos{p,\{q\}}$ would be winning for the attacker);
this excludes in particular $Q' = \varnothing$.
If $Q'$ is a singleton, say $Q' = \{q'\}$, we get $(p',q') \in \rel R$ and the simulation is proven.
Otherwise, the attacker can take a conjunct challenge to $\defenderPos{p',\{\{q'\}\mid q' \in Q'\}}$.
There must be a move from this position to some $\attackerPos{p',\{q'_0\}}$ that is still in the defender winning region $W_\defenderSubscript$. As $q'_0\in Q'$, we can be sure that $q \step{a} q'_0$, so that $(p',q'_0)\in \rel R$ completes this case.
    \end{itemize}
    If $\gameSpectroscopy[\attackerPos{p_0,\{q_0\}}]$ is won by the defender, this means $\attackerPos{p_0,\{q_0\}}\in W_\defenderSubscript$ and thus $(p_0,q_0)\in\rel R$.
    \item Assume $\rel R$ is a bisimulation relation and $(p_0,q_0) \in \rel R$.
Construct the defender strategy $F(\defenderPos{p,\partition{Q}}) = \{ \attackerPos{p,Q} \mid Q \in \partition{Q} \land \exists q\in Q\ldotp (p, q) \in \rel R\}$.
Starting with $\attackerPos{p_0,\{q_0\}}$ whatever the attacker plays and whatever moves the defender chooses from $F$,
we prove that the play will maintain the invariant that for attacker positions $\attackerPos{p,Q}$ or $\attackerPos{p,Q}^\notand$ there will always be $q\in Q$ such that $(p,q) \in \rel R$,
and similarly for defender positions $\defenderPos{p,\partition{Q}}$ there will always be $Q \in \partition{Q}$ and $q \in Q$ such that $(p,q) \in \rel R$.
    \begin{itemize}
      \item Observation moves: At $\attackerPos{p,Q}$, the attacker moves to $\attackerPos{p',Q'}$ with $p \step a p'$ and $Q'$ reachable by $q\step{a}$. By the invariant, there was $q\in Q$ such that $(p,q) \in \rel R$. As $\rel R$ is a simulation, there must be $q' \in Q'$ with $(p', q') \in \rel R$.
      \item Conjunct challenges: At $\attackerPos{p,Q}$, the attacker moves to $\defenderPos{p, \partition{Q}}$. As $\partition{Q}$ covers $Q$, the invariant is maintained.
      \item Conjunct answers: At $\defenderPos{p,\partition{Q}}$, the defender moves to some $g' \in F(\defenderPos{p, \partition{Q}})$. Because of the invariant, the set cannot be empty. Because of the definition of $F$, the new position must keep the invariant.
      \item Negation moves: At $\attackerPos{p,\{q\}}$, the attacker moves to $\attackerPos{q, \{p\}}$. Due to $\rel R$ being symmetric, the invariant holds.
    \end{itemize}
    Because of the invariant, the defender cannot get stuck following $F$, and thus wins.\qedhere
  \end{itemize}
\end{proof}

\subsection{Deciding the Spectroscopy Game\label{subsec:deciding-game}}

Due to the partition construction over subsets of $\proc$,
the worst-case game size is proportional to the Bell number $\bellNumber{1 + \relSize{\proc}}$,
which is somewhat worse than exponential.
Going at least exponential is necessary, as we want to also characterize weaker preorders like the trace preorder, where exponential $\proc$-subset or $\actions^*$-word constructions cannot be circumvented. However, for moderate real-world systems, such constructions will not necessarily show their full exponential blow-up (cf.~\cite{ch1993}).
In our case, we went beyond exponential complexity in order to correctly address failure trace and ready trace languages, as will be explained in \refSubsec{subsec:difference-tacas}.

To be more precise, the game size is asymptotically described by the number of conjunct answers, $\relSize{\proc}^2 \cdot \bellNumber{1 + \relSize{\proc}}$. In detail: There are $\relSize{\proc} \cdot 2^{\relSize{\proc}}$ potential attacker positions. For every such position, there are Bell-number many conjunct challenges and corresponding defender positions (so $\relSize{\proc} \cdot \bellNumber{1 + \relSize{\proc}}$ in total, each) with up to $\relSize{\proc}^2 \cdot \bellNumber{1 + \relSize{\proc}}$ conjunct answers. Moreover, there are up to $\relSize{p \step{\cdot} \cdot}$ observation moves for each attacker position, totalling to $\relSize{\step{\cdot}} \cdot 2^{\relSize{\proc}}$. Finally, there are the $\relSize{\proc}^2$ negation moves, which clearly are dominated by the other complexities.

For our approach in the next section, we will need to know (and process) the whole attacker winning region of size $\bigo(\relSize{\proc} \cdot \bellNumber{1 + \relSize{\proc}})$.
We use Algorithm~\ref{alg:game-algorithm} to compute the winning region of reachability games (derived from~\cite{graedel2007finite}).
It starts by assuming that every position is won by the defender,
and then proceeds by visiting positions where the defender does not have any options or where the attacker has an option.
Every move in the game will lead to at most one invocation of the inner-most operations,
which renders the algorithm linear-time with respect to game moves.
The algorithm keeps information on every game position,
making the space complexity linear in the number of game positions and thus Bell-number super-exponential with respect to $\relSize{\proc}$.

\begin{algorithm}[t]
  \Fn{$\varname{compute\_winning\_region}(\game=(G,G_\defenderSubscript,\gameMove))$}{

    $\varname{defender\_options}\, := [g\mapsto n\mid g\in G\land n=\relSize{\{p\mid g\gameMove p\}}]$

    $\varname{attacker\_win}\, := \{\}$

    \Fn{$\varname{propagate\_attacker\_win}(g)$}{

      \If(){$g \notin \varname{attacker\_win}$}{

        $\varname{attacker\_win} := \varname{attacker\_win} \cup \{g\}$

        \For(){$g_{p}\in (\cdot\gameMove g)$}{

          $\varname{defender\_options}[g_{p}]:=\varname{defender\_options}[g_{p}]-1$

          \If(){$g_{p}\in G_\attackerSubscript\lor\varname{defender\_options}[g_{p}]=0$}{

            $\varname{propagate\_attacker\_win}(g_{p})$

          }
        }
      }
    }

    \For(){$g\in G_\defenderSubscript$}{

      \If(){$\varname{defender\_options}[g]=0$}{

        $\varname{propagate\_attacker\_win}(g)$

      }
    }

    \KwRet{$\varname{attacker\_win}$}

  }
  \caption{\label{alg:game-algorithm}Algorithm for deciding the attacker winning region $W_\attackerSubscript$ of a reachability game $\game$ in linear time of $\relSize{\gameMove}$ and linear space of $\relSize{G}$.}
\end{algorithm}

\subsection{Building Distinguishing Formulas from Attacker Strategies\label{subsec:building-formulas}}

Attacker strategies in $\gameSpectroscopy$ correspond to sets of $\hml$ formulas. If the strategies are winning, the formulas are distinguishing.

\begin{defi}[Strategy formulas]%
  \label{def:strategy-formulas}
  Given a (positional, non-deterministic, labeled) attacker strategy $F \subseteq (G_\attackerSubscript \times L \times G) \cap \mathord{\gameMove}$ for the spectroscopy game $\gameSpectroscopy$,
  the set of \emph{strategy formulas,} $\hmlStrategies_F$, is inductively defined by:

  \begin{itemize}
      \item If $\varphi \in \hmlStrategies_F(g_\attackerSubscript^\prime)$ and $(g_\attackerSubscript, \hmlObs{b}, g_\attackerSubscript^\prime) \in F$, then $\hmlObs{b}\varphi \in \hmlStrategies_F(g_\attackerSubscript)$,

      \item if $\varphi \in \hmlStrategies_F(g_\attackerSubscript^\prime)$ and $(g_\attackerSubscript, \hmlNeg, g_\attackerSubscript^\prime) \in F$, then $\hmlNeg\varphi \in \hmlStrategies_F(g_\attackerSubscript)$,

      \item if $\varphi \in \hmlStrategies_F(g_\defenderSubscript)$ and $(g_\attackerSubscript, \wedge, g_\defenderSubscript) \in F$, then $\varphi \in \hmlStrategies_F(g_\attackerSubscript)$, and

      \item if $\varphi_{g_\attackerSubscript'} \in \hmlStrategies_F(g_\attackerSubscript')$ for all $g_\attackerSubscript' \in I = \{g_\attackerSubscript' \mid g_\defenderSubscript \gameMoveX{*}_\triangle g_\attackerSubscript'\}$ then $\hmlAnd{g_\attackerSubscript'}{I}\varphi_{g_\attackerSubscript'}\in \hmlStrategies_F(g_\defenderSubscript)$.
  \end{itemize}
\end{defi}

\begin{exa}
  The attacks $\attackerPos{\ccsIdentifier{P_1},\{\ccsIdentifier{P_2}\}}
  \gameMoveX{\hmlObsI{a}} \attackerPos{\ccsInm{b}\ccsChoice\ccsInm{c},
      \{\ccsInm{b}\ccsChoice\ccsInm{d}, \ccsInm{c}\ccsChoice\ccsInm{d}\}}
  \gameMoveX{\land} \gameMoveX{*}
  \gameMoveX{\hmlNeg} \gameMoveX{\hmlObsI{d}}
  \attackerPos{\ccsStop,\varnothing}\gameMoveX{\land}$ on the system of \refExample{ex:ccs-p1p2} give rise to the formula
  $\hmlObsI{a}\hmlAndS\{\hmlNeg\hmlObsI{d}\hmlTrue\}$, which can be written as $\hmlObsI{a}\hmlNeg\hmlObsI{d}$.
\end{exa}

The definition will never generate disjunctive subformulas:
The target of a negation move is always a position $\attackerPos{p,\{q\}}$ with singleton $\{q\}$,
and only observation moves and negation moves are generated in such positions.

\begin{defi}[Winning strategy graph]%
  \label{def:strategy-graph}
  Given the attacker winning region $W_\attackerSubscript$ and a starting position $g_0 \in W_\attackerSubscript$,
  the \emph{attacker winning strategy graph} $F_\attackerSubscript$ is the subset of the $\gameMove\:\!$-graph
  that can be visited from $g_0$ when following all $\gameMove\:\!$-edges unless they lead out of $W_\attackerSubscript$.
\end{defi}

This graph can be cyclic;
cycles are of the form $\attackerPos{p,\{q\}} \gameMoveX{\hmlNeg} \attackerPos{q,\{p\}} \gameMoveX{\hmlNeg} \attackerPos{p,\{q\}}$,
or introduced by recursive processes.
(Remember that our examples are acyclic, but more complex LTSs could contain cycles.)
However, if the attacker plays inside their winning region according to $F_\attackerSubscript$,
they will always have paths to their final winning positions.
So even though the attacker could loop (and thus introduce double negations or lose),
they can always end the game and \emph{win} in the sense of \refDef{def:strategies}
without generating double negations.

\newsavebox{\canceledbox}\sbox{\canceledbox}{$\cancel{\text{\normalsize\color{black}pruned away}}$}

\begin{figure}[tph!]
  \vspace*{0.8cm} 
  \begin{center}
  \begin{tikzpicture}[shorten <=1pt,shorten >=0.5pt]

    \path[use as bounding box,draw=none] (8.1,-5.44) rectangle (-7.05,10.91);

    \begin{scope}[every node/.style={inner sep=1ex,anchor=center,draw,black,fill=gray!5,thick}]

    \node(s0) at(-2.5,10) {\begin{tabular}{@{\;}c@{\;}}$
    \attackerPos{\ccsInm{a}\ccsPrefix\ccsInm{b},
      \{ \ccsInm{a}\ccsPrefix(\ccsInm{a}\ccsChoice \ccsInm{b}) \ccsChoice
        \ccsInm{a}\ccsPrefix\ccsInm{b}\ccsPrefix\ccsInm{b} \ccsChoice
        \ccsInm{a} \}}
  $\\[.1333em]\hline\raisebox{-.1333em}{$
    \hmlObs{a}\hmlAndS\{\hmlNeg\hmlObs{a}, \hmlObs{b}\hmlNeg\hmlObs{b}\};
    \text{\color{black!75}\small$\hmlObs{a}\hmlAndS\{\hmlNeg\hmlObs{a}, \hmlNeg\hmlObs{b}\hmlObs{b}, \hmlObs{b}\}$};
  $}\\[.1333em]$
    \hmlNeg\hmlObs{a}\hmlObs{a};
    \cancel{\hmlNeg\hmlNeg\hmlObs{a}\hmlAndS\{\hmlNeg\hmlObs{a}, \hmlNeg\hmlObs{b}\hmlObs{b}, \hmlObs{b}\}}
  $\end{tabular}};

    \node(s0neg) at(5.1,8.5) {\begin{tabular}{@{\;}c@{\;}}$
    \attackerPos{\ccsInm{a}\ccsPrefix(\ccsInm{a}\ccsChoice
      \ccsInm{b}) \ccsChoice \ccsInm{a}\ccsPrefix\ccsInm{b}\ccsPrefix\ccsInm{b} \ccsChoice
      \ccsInm{a},
      \{ \ccsInm{a}\ccsPrefix\ccsInm{b} \}}
  $\\[.1333em]\hline\raisebox{-.1333em}{$
    \hmlObs{a}\hmlObs{a};
    \cancel{\hmlObs{a}\hmlObs{b}\hmlObs{b}};
    \cancel{\hmlObs{a}\hmlNeg\hmlObs{b}\hmlNeg\hmlObs{b}};
  $}\\[.1333em]$
    \cancel{\hmlObs{a}\hmlNeg\hmlObs{b}};
    \cancel{\hmlNeg\hmlObs{a}\hmlAndS\{\hmlNeg\hmlObs{a}, \hmlObs{b}\hmlNeg\hmlObs{b}\}};
  $\\$
    \text{\color{black!75}\small$\hmlNeg\hmlObs{a}\hmlAndS\{\hmlNeg\hmlObs{a}, \hmlNeg\hmlObs{b}\hmlObs{b}, \hmlObs{b}\}$};
    \cancel{\hmlNeg\hmlNeg\hmlObs{a}\hmlObs{a}}
  $\end{tabular}};

    \node(s1) at(-2.5,7.25) {\begin{tabular}{@{\;}c@{\;}}$
    \attackerPos{\ccsInm{b},
      \{ \ccsInm{a}\ccsChoice \ccsInm{b},
         \ccsInm{b}\ccsPrefix\ccsInm{b},
         \ccsStop \}}
  $\\[.1333em]\hline\raisebox{-.1333em}{$
    \hmlAndS\{\hmlNeg\hmlObs{a}, \hmlObs{b}\hmlNeg\hmlObs{b}\};
  $}\\[.1333em]$
    \cancel{\hmlAndS\{\hmlNeg\hmlObs{a}, \hmlObs{b}\hmlNeg\hmlObs{b}, \hmlObs{b}\}};
    \text{\color{black!75}\small$\hmlAndS\{\hmlNeg\hmlObs{a}, \hmlNeg\hmlObs{b}\hmlObs{b}, \hmlObs{b}\}$}
  $\end{tabular}};

    \node[ellipse,minimum height=10ex,minimum width=12em,text height=2ex](s2a) at(-4.5,4.75) {\smash{\makebox[0pt]{\begin{tabular}{c}$
    \defenderPos{\ccsInm{b},
      \{ \{\ccsInm{a}\ccsChoice \ccsInm{b}\},
         \{\ccsInm{b}\ccsPrefix\ccsInm{b}, \ccsStop\} \}}
  $\\[.1333em]\hline\raisebox{-.1333em}{$
    \hmlAndS\{\hmlNeg\hmlObs{a}, \hmlObs{b}\hmlNeg\hmlObs{b}\}
  $}\end{tabular}}}};

    \node(s3b) at(-2.75,2.25) {\begin{tabular}{@{\;}c@{\;}}$
    \attackerPos{\ccsInm{b}, \{ \ccsInm{b}\ccsPrefix\ccsInm{b}, \ccsStop \}}^\notand
  $\\[.1333em]\hline\raisebox{-.1333em}{$
    \hmlObs{b}\hmlNeg\hmlObs{b}
  $}\end{tabular}};

    \node(s3c) at(4,0.5) {\begin{tabular}{@{\;}c@{\;}}$
    \attackerPos{\ccsInm{b}, \{ \ccsStop \}}
  $\\[.1333em]\hline\raisebox{-.1333em}{$
    \hmlObs{b};
    \cancel{\hmlNeg\hmlNeg\hmlObs{b}}
  $}\end{tabular}};

    \node(s3a) at(-5.75,0.5) {\begin{tabular}{@{\;}c@{\;}}$
    \attackerPos{\ccsInm{b}, \{ \ccsInm{a}\ccsChoice \ccsInm{b} \}}
  $\\[.1333em]\hline\raisebox{-.1333em}{$
    \hmlNeg\hmlObs{a}
  $}\end{tabular}};

    \node(s4c) at(4,-1.5) {\begin{tabular}{@{\;}c@{\;}}$
    \attackerPos{\ccsStop, \{ \ccsInm{b} \}}
  $\\[.1333em]\hline\raisebox{-.1333em}{$
    \hmlNeg\hmlObs{b}
  $}\end{tabular}};

    \node(s4a) at(-5.75,-1.5) {\begin{tabular}{@{\;}c@{\;}}$
    \attackerPos{\ccsInm{a}\ccsChoice \ccsInm{b}, \{ \ccsInm{b} \}}
  $\\[.1333em]\hline\raisebox{-.1333em}{$
    \hmlObs{a};
    \cancel{\hmlNeg\hmlNeg\hmlObs{a}}
  $}\end{tabular}};

    \node(s5) at(-0.75,-4) {\begin{tabular}{@{\;}c@{\;}}$
    \attackerPos{\ccsStop, \varnothing}
  $\\[.1333em]\hline\raisebox{-.1333em}{$
    \hmlTrue
  $}\end{tabular}};

    \node[ellipse,minimum height=8ex,inner sep=2pt,minimum width=5.6em,text height=1.5ex](bott) at(2.75,-4) {\smash{\makebox[0pt]{\begin{tabular}{c}$
    \defenderPos{\ccsStop, \varnothing}
  $\\[.1333em]\hline\raisebox{-.1333em}{$
    \hmlTrue
  $}\end{tabular}}}};

    \end{scope}


    \begin{scope}[every node/.style={inner sep=1ex,anchor=center,draw,black!75,fill=gray!4,node font=\small}]

    \node[ellipse,minimum height=15ex,minimum width=13em,text height=2.5ex](s2b) at(2,4.75) {\smash{\makebox[0pt]{\begin{tabular}{c}$
    \defenderPos{\ccsInm{b},
      \{ \{\ccsInm{a}\ccsChoice \ccsInm{b}\},
         \{\ccsInm{b}\ccsPrefix\ccsInm{b}\},
         \{\ccsStop\} \}}
  $\\[.1333em]\hline\raisebox{-.1333em}{$
    {\phantom{;}}\hmlAndS\{\hmlNeg\hmlObs{a}, \hmlObs{b}\hmlNeg\hmlObs{b}, \hmlObs{b}\}{;}
  $}\\[.1333em]$
    \hmlAndS\{\hmlNeg\hmlObs{a}, \hmlNeg\hmlObs{b}\hmlObs{b}, \hmlObs{b}\}
  $\end{tabular}}}};

    \node(s3d) at(-0.5,0.5) {\begin{tabular}{@{\;}c@{\;}}$
    \attackerPos{\ccsInm{b}, \{ \ccsInm{b}\ccsPrefix\ccsInm{b} \}}
  $\\[.1333em]\hline\raisebox{-.1333em}{$
    \hmlObs{b}\hmlNeg\hmlObs{b};
    \hmlNeg\hmlObs{b}\hmlObs{b};
    \cancel{\hmlNeg\hmlNeg\hmlObs{b}\hmlNeg\hmlObs{b}}
  $}\end{tabular}};

    \node(s4d) at(-0.5,-1.5) {\begin{tabular}{@{\;}c@{\;}}$
    \attackerPos{\ccsInm{b}\ccsPrefix\ccsInm{b}, \{ \ccsInm{b} \}}
  $\\[.1333em]\hline\raisebox{-.1333em}{$
    \hmlObs{b}\hmlObs{b};
    \hmlNeg\hmlObs{b}\hmlNeg\hmlObs{b};
    \cancel{\hmlNeg\hmlNeg\hmlObs{b}\hmlObs{b}}
  $}\end{tabular}};

    \end{scope}


    \begin{scope}[>->,black!75,every node/.style={node font=\small}]

    \draw ([yshift=5mm]s0neg) -- node[above]{$~~\hmlNeg$} (s0);

    \draw (s0neg) .. controls (6.75,3) and (8.25,-5.375) .. node[right,pos=.45]{$\hmlObs{a}$} (s4d);
    \draw (s0neg.275) .. controls (6,3) and (6.75,0.75) .. node[left,pos=.56]{$\hmlObs{a}$} (s4c);

    \draw (s1) -- node[right,pos=.4]{$~~\land$} (s2b);
    \draw (s2b) .. controls (-4.625,3.25) and (-4,3.875) .. node[above,pos=.1]{$*~$} (s3a);
    \draw (s2b) -- node[right]{$*$} (s3c);
    \draw (s2b) -- node[left]{$*$} (s3d);

    \draw (s4a.77) -- node[right]{$\hmlNeg$} (s3a.283);
    \draw (s3c.257) -- node[left]{$\hmlNeg$} (s4c.103);
    \draw (s3d.257) -- node[left]{$\hmlNeg$} (s4d.103);
    \draw (s4d.77) -- node[right]{$\hmlNeg$} (s3d.283);

    \draw (s3d.345) -- node[below,pos=.65]{$\hmlObs{b}~$} (s4c);
    \draw (s4d.16) -- node[above,pos=.725]{$\hmlObs{b}~$} (s3c);

    \end{scope}


    \begin{scope}[>->,black,thick,every node/.style={node font=\boldmath}]

    \draw (s0) -- node[left]{$\hmlObs{a}$} (s1);
    \draw[shorten <=1.5pt] ([yshift=-5mm]s0) -- node[below]{$\hmlNeg~~$} (s0neg);

    \draw (s0neg.297) .. controls (11.75,-7.5) and (-2.75,-7.5) .. node[right,pos=.22]{$\hmlObs{a}$} (s4a);

    \draw (s1) -- node[left,pos=.35]{$\land~$} (s2a);
    \draw (s2a) -- node[left,pos=.23]{$*$} (s3a);
    \draw (s2a) -- node[right,pos=.65]{$*$} (s3b);

    \draw (s4c.77) -- node[right]{$\hmlNeg$} (s3c.283);
    \draw (s3a.257) -- node[left]{$\hmlNeg$} (s4a.103);

    \draw (s3b) .. controls (-4.5,-2.75) and (-2.25,-3.75) .. node[left,pos=.06]{$\hmlObs{b}$} (s4c);
    \draw (s4a) -- node[below,pos=.45]{$\hmlObs{a}~~$} (s5);
    \draw (s3c) .. controls (7.5,-3.25) and (1.75,-2.75) .. node[right,pos=.1]{$\hmlObs{b}$} (s5.20);
    \draw (s5) -- node[below]{$\land$} (bott);

    \end{scope}

  \end{tikzpicture}
  \end{center}
  \vspace*{0.5cm}
  \caption{\label{fig:spectro-game-formulas}%
  The winning region of the spectroscopy game of \refFig{fig:spectro-game-1},
  together with distinguishing formulas.
  Every node displays in the top half the game position
  and in the bottom half the relevant cheap strategy formulas in $\hmlStrategies_{F_\attackerSubscript}$ for that position.
  Formulas that are recognized as too expensive to be interesting are \usebox{\canceledbox},
  as will be explained more formally in \refDef{def:formula-pruning}.
  Edges that are used to construct the main distinguishing formulas are {\boldmath $\mathit{thick.}$}}
  \vspace*{0.5cm}
\end{figure}

\begin{exa}\label{exa:game-distinguishing-formulas}
  \refFig{fig:spectro-game-formulas} continues the comparison of $\ccsInm{a}\ccsPrefix\ccsInm{b}$ with $\ccsInm{a}\ccsPrefix(\ccsInm{a}\ccsChoice \ccsInm{b}) \ccsChoice \ccsInm{a}\ccsPrefix\ccsInm{b}\ccsPrefix\ccsInm{b} \ccsChoice \ccsInm{a}$
  from \refExample{exa:spec-game-1}.
  The game constructs the following three distinguishing formulas:

  \begin{itemize}
  \item	$\hmlObs{a}\hmlAndS\{\hmlNeg\hmlObs{a}, \hmlObs{b}\hmlNeg\hmlObs{b}\}$.
    This formula has price $(3,2,1,1,1,1)$ and shows that the processes can be distinguished by a failure trace,
    namely $\textcolor{red}{\varnothing} \ccsInm{a} \textcolor{red}{\{ \ccsInm{a} \}} \ccsInm{b} \textcolor{red}{\{ \ccsInm{b} \}}$:
    this is a failure trace of $\ccsInm{a}\ccsPrefix\ccsInm{b}$ but not of $\ccsInm{a} \ccsChoice \ccsInm{a}\ccsPrefix\ccsInm{b}\ccsPrefix\ccsInm{b} \ccsChoice \ccsInm{a}\ccsPrefix(\ccsInm{a} \ccsChoice \ccsInm{b})$.
  \item	$\hmlNeg \hmlObs{a}\hmlObs{a}$,
    with price $\expr(\widehat{\,\cdot\,}) = (2,1,0,0,1,2)$.
    This formula shows that the processes are not in the impossible-futures preorder.
  \item $\hmlObs{a}\hmlAndS\{\hmlNeg\hmlObs{a}, \hmlNeg\hmlObs{b}\hmlObs{b}, \hmlObs{b}\}$.
    This formula has price $(3,1,0,1,1,2)$ and shows that the processes can be distinguished by the possible-futures preorder.
    While this formula is strictly more expensive than $\hmlNeg\hmlObs{a}\hmlObs{a}$
    (and therefore will be suppressed in the final output),
    it cannot be deleted altogether
    because it may be needed as part of a larger formula in other states.
    In particular, its negation is strictly cheaper than the negation of the first constructed formula,
    similar to \refExample{exa:locally-more-expressive}.
  \end{itemize}
  The construction of other formulas is aborted
  as soon as they are recognized as more expensive in every relevant context,
  as we will discuss in the next subsection.
\end{exa}

Let us quickly prove the soundness of the strategy formulas obtained from winning strategies.

\begin{thm}%
  \label{thm:spectroscopy-hml}
  If $W_\attackerSubscript$ is the attacker winning region of the spectroscopy game $\gameSpectroscopy$ and $F_\attackerSubscript$ the derived strategy graph, every $\varphi \in \hmlStrategies_{F_\attackerSubscript}(\attackerPos{p, Q})$ distinguishes $p$ from every $q \in Q$.
\end{thm}
\begin{proof}
  We proceed by induction on the structure of $\hmlStrategies_{F_\attackerSubscript}$ with arbitrary $p, Q$.
  \begin{itemize}
    \item Assume $\varphi \in \hmlStrategies_{F_\attackerSubscript}(\attackerPos{p', Q'})$, and $(\attackerPos{p, Q},\hmlObs{b}, \attackerPos{p', Q'}) \in F_\attackerSubscript$.
    By induction hypothesis, $\varphi$ must be true for $p'$ and false for all $q' \in Q'$.
    Due to the structure of $\gameSpectroscopy$, we know that $p \step{b} p'$ and that $Q'$ contains all $b$-reachable states for $Q$.
    Thus, $\hmlObs{b}\varphi \in \hmlStrategies_{F_\attackerSubscript}(\attackerPos{p, Q})$ must be true for $p$ and false for all $q \in Q$. Likewise, in the case of $(\attackerPos{p, Q}^\notand,\hmlObs{b}, \attackerPos{p', Q'}) \in F_\attackerSubscript$, we find that $\hmlObs{b}\varphi \in \hmlStrategies_{F_\attackerSubscript}(\attackerPos{p, Q}^\notand)$ is distinguishing.
    \item Assume $\varphi \in \hmlStrategies_{F_\attackerSubscript}(\attackerPos{p', Q'})$ and $(\attackerPos{p, Q}, \hmlNeg, \attackerPos{p', Q'}) \in F_\attackerSubscript$.
    By the construction of $\gameSpectroscopy$, $Q=\{p'\}$ and $Q'=\{p\}$.
    By induction hypothesis, $\varphi$ must be true for $p'$ and false for $p$.
    So, $\hmlNeg\varphi \in \hmlStrategies_{F_\attackerSubscript}(\attackerPos{p, Q})$ must be true for $p$ and false for all elements of $\{p'\} = Q$.
    \item Assume $\varphi_{g_\attackerSubscript'} \in \hmlStrategies_{F_\attackerSubscript}(g_\attackerSubscript')$ for all $g_\attackerSubscript' \in I = \{g_\attackerSubscript' \mid g_\defenderSubscript \gameMoveX{*}_\triangle g_\attackerSubscript'\}$, and also $(\attackerPos{p, Q},\wedge, g_\defenderSubscript) \in F_\attackerSubscript$.
    Due to the construction of $\gameSpectroscopy$, $Q=\{ q' \mid \attackerPos{p', \{q'\}} \in I \} \cup \{q' \mid \exists Q'\ldotp q' \in Q' \land \attackerPos{p', Q'}^\notand  \in I\}$ and $p'=p$.
    By induction hypothesis, every $\varphi_{g_\attackerSubscript'}$ is true for $p$ and false for all its respective $q' \in Q'$.
    So, the conjunction $\hmlAnd{g_\attackerSubscript'}{I}\varphi_{g_\attackerSubscript'} \in \hmlStrategies_{F_\attackerSubscript}(\attackerPos{p, Q})$ must be distinguishing for $p$ and~$Q$.
    \qedhere
  \end{itemize}
\end{proof}

\noindent
Note that the theorem is only one-way, as every distinguishing formula can neutrally be extended by stating that some additional clause that is true for \emph{both} processes does hold. \refDef{def:strategy-formulas} will not find such bloated formulas.

We prove that a neatly constructed subset of the cheapest strategy formulas suffices for reaching completeness in \refThm{thm:correctness-general}.

Due to the cycles in the game graph, $\hmlStrategies_{F_\attackerSubscript}$ usually yields infinitely many formulas, unbounded in length and price.
The next section will discuss how to find a finite subset of $\hmlStrategies_{F_\attackerSubscript}$ that solves Problem~\ref{prob:cheapest-distinction}.

\subsection{\label{subsec:cheapest-formulas}Retrieving Cheapest Distinguishing Formulas}

We are now at the point where the recharting of the spectrum from \refSubsec{subsec:priced-formulas} pays off.
With the pricing metric on formulas, the coarsest ways of telling two states $p,q$ apart are precisely given by the formulas $\Phi_\Delta \subseteq \hmlStrategies_{F_\attackerSubscript}(\attackerPos{p,\{q\}})$ that are not dominated by any other formula in the set.

There might well be multiple such minimal-price formulas $\varphi_\Delta \in \Phi_\Delta$
with differing $\expr$ prices for the same pair of processes.
Due to \refLem{lem:obslang-price-characterization}, the processes then are distinguished with respect to every equivalence $X$
where the budget vector $e_X$ from Table~\ref{tab:ltbts-dimensions} is above the price of one of the formulas,
that is, where $\expr(\widehat{\varphi_\Delta}) \sqsubseteq e_X$ for a $\varphi_\Delta \in \Phi_\Delta$.
At the same time, the processes are preordered with respect to all the other notions of equivalence from the table,
because these do not have sufficient distinguishing capabilities.

\begin{algorithm}[t]
  \Fn{\upshape $\varname{game\_spectroscopy}$($\system,p_0,q_0$)}{

    $\gameSpectroscopy^{\system}=(G,G_\attackerSubscript,\gameMove) := \varname{construct\_spectroscopy\_game}(\system)$%
    \label{algo:line:construct-spectroscopy-game}

    $W_\attackerSubscript:=\varname{compute\_winning\_region}(\gameSpectroscopy^{\system})$

    \If(){$\attackerPos{p_0,\{q_0\}} \in W_\attackerSubscript$}{

      $F_\attackerSubscript := \varname{winning\_graph}(\gameSpectroscopy^{\system}, W_\attackerSubscript, \attackerPos{p_0,\{q_0\}})$%
      \label{algo:line:begin-distinguish-formula}

      $\varname{strats}{[]} := \varnothing$

      $\varname{todo} := [\attackerPos{p_0,\{q_0\}}]$

      \While(){
        $\varname{todo} \neq []$
      }{
        $\varname{g} := \varname{todo.dequeue}()$

        $\varname{sg} :=  \varname{strats}{[\varname{g}]}$

        \If(){$\varname{sg} = \textit{undefined}$}{
          $\varname{strats}{[\varname{g}]} := \varnothing$
        }

        $\varname{gg'} := \{g' \mid (\varname{g},\cdot, g') \in F_\attackerSubscript \land \varname{strats}[g'] = \textit{undefined}\}$

        \If(){$\varname{gg'} = \varnothing$}{
          $\varname{sg'} = \hmlPrune(\hmlStrategies'_{F_\attackerSubscript,\varname{strats}}(\varname{g}))$

          \If(){$\varname{sg} \neq \varname{sg'}$}{
            $\varname{strats}[\varname{g}] := \varname{sg'}$

            $\varname{todo.enqueue\_each\_end}( \{g^* \mid (g^*,\cdot, \varname{g}) \in F_\attackerSubscript \land g^* \notin \varname{todo} \} )$
          }

        }
        \Else(){
          $\varname{todo.enqueue\_each\_front}(\varname{gg'})$%
          \label{algo:line:end-distinguish-formula}
        }
      }
      \KwRet{$\varname{strats}[\attackerPos{p_0,\{q_0\}}]$}
    }

    \Else(){

      $\rel R :=\{(p,q)\mid \attackerPos{p,\{q\}} \in G_\attackerSubscript \setminus W_\attackerSubscript \}$

      \KwRet{$\rel R$}

    }

  }\caption{\label{alg:spectroscopy-algo}Spectroscopy procedure.}
\end{algorithm}

The extraction of cheapest formulas from $\hmlStrategies_{F_\attackerSubscript}$ presents itself
as a very special kind of ``shortest-distance'' problem (cf.~\cite{mohri2002semiringShortest})
between the positions immediately won by the attacker and the initial game position.
It is just that the ``distances'' we have to keep track of are no numbers but sets of formulas.
Those are peculiar as we have seen in \refExample{exa:locally-more-expressive}.
We cannot really calculate with them but only compare their $\expr$ values and construct more complex formula sets.

The rest of this subsection is about designing a \emph{fixed point algorithm} around this strange problem space.
In order to be efficient, the algorithm should discard distinguishing formulas
that are already recognized as not-cheapest as early as possible.

Algorithm~\ref{alg:spectroscopy-algo} gives an overview of how the results of previous subsections and two following definitions play together.
It constructs the spectroscopy game $\gameSpectroscopy^\system$ (\refDef{def:spectroscopy-game}),
finds the attacker winning region using Algorithm~\ref{alg:game-algorithm}
and computes its attacker winning strategy graph $F_\attackerSubscript$ (\refDef{def:strategy-graph}).
If the attacker cannot win, the algorithm returns a bisimulation relation derived from the defender winning region (Lemma~\ref{lem:spectroscopy-game-bisimulation}).
Otherwise, it constructs the distinguishing formulas:
It keeps a map $\varname{strats}$ of strategy formulas that have been \emph{found so far}
and a list of game positions $\varname{todo}$ that have to be updated.
In every round, we take a game position $\varname{g}$ from $\varname{todo}$.
If some of its successors have not been visited yet, we add them to the top of the work list.
Otherwise we call $\hmlStrategies'_{F_\attackerSubscript,\varname{strats}}(\varname{g})$ (to be introduced in \refDef{def:partial-strategy-formulas})
to compute distinguishing formulas using the follow-up formulas $\varname{strats}$.
From the found formulas, $\hmlPrune$ removes formulas that are recognized as too expensive to be interesting
(to be introduced in \refDef{def:formula-pruning}).
If the result changes $\varname{strats}[\varname g]$,
we enqueue each game predecessor to propagate the update there.

The fixed point approach needs $\hmlStrategies_{F}$ from \refDef{def:strategy-formulas} in a more stepwise formulation that is not recursive.
The function $\hmlStrategies'_{F,\mathit{strats}}$ mostly corresponds to \refDef{def:strategy-formulas}
with the twist that tentative follow-ups are used instead of recursion.

\begin{defi}[Tentative strategy formulas]%
  \label{def:partial-strategy-formulas}
  Given a labeled attacker strategy $F \subseteq (G_\attackerSubscript \times L \times G)$
  for the spectroscopy game $\gameSpectroscopy$
  and an approximation of strategy formulas $\mathit{strats} \colon G \rightarrow \hmlA$,
  the \emph{next tentative strategy formulas,} $\hmlStrategies'_{F,\mathit{strats}}$, are defined for $g_\attackerSubscript \in G_\attackerSubscript$ and $g_\defenderSubscript \in G_\defenderSubscript$ as:
  \begin{align*}
    \hmlStrategies'_{F,\mathit{strats}}(g_\attackerSubscript) =
    & \phantom{{} \cup {}}\{ \hmlObs{b}\varphi \mid (g_\attackerSubscript, \hmlObs{b}, g_\attackerSubscript^\prime) \in F \land \varphi \in \mathit{strats}(g_\attackerSubscript^\prime)\}\\
    & \cup \{ \hmlNeg\varphi \mid (g_\attackerSubscript, \hmlNeg, g_\attackerSubscript^\prime) \in F \land \varphi \in \mathit{strats}(g_\attackerSubscript^\prime)\}\\
    & \cup \{\varphi \mid (g_\attackerSubscript, \wedge, g_\defenderSubscript) \in F \land \varphi \in \mathit{strats}(g_\defenderSubscript)\}\\
    \hmlStrategies'_{F,\mathit{strats}}(g_\defenderSubscript) =
    & \phantom{{} \cup {}}\{ \textstyle\hmlAnd{g_\attackerSubscript}{I}\varphi_{g_\attackerSubscript} \mid I = \{g_\attackerSubscript \mid g_\defenderSubscript \gameMoveX{*}_\triangle g_\attackerSubscript \} \land \forall g_\attackerSubscript \in I \ldotp \varphi_{g_\attackerSubscript} \in \mathit{strats}(g_\attackerSubscript)\}
  \end{align*}
\end{defi}

Intuitively, $\mathit{strats}$ corresponds to intermediate candidates for reaching the target region in a shortest-distance problem.
$\hmlStrategies'_{F,\mathit{strats}}$ corresponds to the addition of outgoing edge weights to the candidates when updating nodes.
For a shortest-distance algorithm, we need one more ingredient: the criterion to select best candidates.
In conventional shortest-distance problems, the minimum function ($\min$) plays this role.
$\hmlStrategies'$ already implicitly does something comparable to $\min$: It forms the union of all the formula sets that are implied by outgoing game moves.
This union on its own has the problem that formulas may be unfolded beyond all bounds,
accumulating more and more information in arbitrarily expensive distinguishing formulas.
For the fixed point algorithm to terminate, there must be a pointwise bound on the growth of $\mathit{strats}$ with only finitely many possible sets of formulas below.
For instance, one could restrict the formulas to have expressiveness prices below $(\relSize{S},\relSize{S},\relSize{S},\relSize{S},\relSize{S},\relSize{S})$.
This would guarantee a safe solution, but it can result in insane running times.
Instead, we use the following function:

\begin{defi}[Pruning of dominated formulas]\label{def:formula-pruning}
  $\hmlPrune$ is defined as
  \begin{align*}
    \hmlPrune(\obs{})  = \{ \varphi \in \obs{} \mid
       & \phantom{{}\land{}}(\nexists \varphi' \ldotp \varphi = \hmlNeg\hmlNeg \varphi') \\
       & \land (\nexists \psi \in \obs{} \ldotp \expr(\widehat{\psi}) \sqsubset \expr(\widehat{\vphantom{\psi}\varphi}) \\
       & \qquad \land ((\exists a, \varphi'\ldotp \varphi = \hmlObs{a}\varphi')
          \rightarrow (\exists a, \psi'\ldotp \psi = \hmlObs{a}\psi')) \\
      & \qquad \land ((\exists \varphi'\ldotp \varphi = \hmlNeg\varphi')
          \rightarrow (\exists \psi'\ldotp \psi = \hmlNeg\psi'))
      )\}.
  \end{align*}
\end{defi}

$\hmlPrune$ removes all formulas that are dominated with respect to the metrics by any other formula in this set,
where observations may only be dominated by observations, and negations only by negations
(but double negations are always dropped because we forbid them in distinguishing formulas).
We need to keep minimum-price negations because of inversion of dominance in cases like \refExample{exa:locally-more-expressive}.
We also need to keep minimum-price observation formulas,
at least in game positions $\attackerPos{p,\{q\}}$,
because only such formulas are allowed under a negation.
Other dominated formulas will not be important later on to find the cheapest options.

\begin{exa}
  If we look back at \refFig{fig:spectro-game-formulas} (used for \refExample{exa:game-distinguishing-formulas}),
  we can now see why it is important that dominated formulas are pruned away \emph{at the right steps.}
  In the figure, pruned formulas are $\cancel{\text{\normalsize\color{black}striked out}}$.

  If a game position $\attackerPos{p,\{q\}}$ allows a positive and a negative distinguishing formula,
  we need to maintain the cheapest ones of either kind.
  For example, we kept both $\hmlObs{b}\hmlNeg\hmlObs{b}$ and the (more expensive) $\hmlNeg\hmlObs{b}\hmlObs{b}$ in $\attackerPos{\ccsInm{b}, \{ \ccsInm{b}\ccsPrefix\ccsInm{b} \}}$,
  which allowed to construct two minimal-price distinguishing formulas in $\defenderPos{\ccsInm{b}, \{ \{ \ccsInm{a}\ccsChoice\ccsInm{b} \}, \{ \ccsInm{b}\ccsPrefix\ccsInm{b} \}, \{ \ccsStop \} \}}$:
  one formula with two positive branches and a small depth of negated observations, the other with one positive branch but deeper negated observations.
  However, one of them is pruned away in $\attackerPos{\ccsInm{b}, \{ \ccsInm{a}\ccsChoice\ccsInm{b}, \ccsInm{b}\ccsPrefix\ccsInm{b}, \ccsStop \}}$
  because the latter position has another successor that generates a strictly cheaper formula.

\end{exa}

All in all, the algorithm structure in Algorithm~\ref{alg:spectroscopy-algo} is mostly usual fixed point machinery. It terminates because, for each state in a finite transition system, there must be a bound on the distinguishing mechanisms necessary with respect to our metrics. $\hmlStrategies'$ will only generate finitely many formulas under this bound. $\hmlPrune$ ensures that not too many more formulas are generated.

\subsection{\label{subsec:algo-correctness}Correctness of the Algorithm}

We will now prove that Algorithm~\ref{alg:spectroscopy-algo}---in spite of the pruning---generates enough formulas to solve Problem~\ref{prob:cheapest-distinction}.

\begin{thm}\label{thm:correctness-general}
  Assume a finite formula $\varphi$
  that distinguishes $p$ from every $q \in Q$.
  Assume also that $\varphi$ does not contain double negations or negated conjunctions
  (i.e.\@ does not have subformulas of the form $\hmlNeg\hmlNeg\psi$ or $\hmlNeg\hmlAnd{i}{I}\psi_i$).
  Then we claim about the sets $\varname{strats}[\attackerPos{p,Q}]$ and $\varname{strats}[\attackerPos{p,Q}^\notand]$ of distinguishing formulas produced by Algorithm~\ref{alg:spectroscopy-algo}:
  \begin{enumerate}
  \item\label{thm:correctness-general-generalclaim}
    $\varname{strats}[\attackerPos{p,Q}]$ contains a formula
    that is cheaper than or as cheap as $\varphi$.
  \item\label{thm:correctness-general-postconjunctionclaim}
    If $\varphi$ is an observation formula,
    $\varname{strats}[\attackerPos{p,Q}^\notand]$ contains a formula
    that is cheaper than or as cheap as $\varphi$.
    (This necessarily is an observation formula.)
  \item\label{thm:correctness-general-singletonclaim}
    If $\varphi$ is an observation formula or a negation formula and $Q$ is a singleton,
    $\varname{strats}[\attackerPos{p,Q}]$ contains a formula of the same kind
    that is cheaper than or as cheap as $\varphi$.
  \end{enumerate}
\end{thm}

\begin{proof}
First, $\varname{strats}[\attackerPos{p,\varnothing}] = \{\hmlTrue\}$,
and $\hmlTrue$ is the cheapest of all formulas, so the claims hold if $Q = \varnothing$.
In the rest of the proof we assume that $Q \not= \varnothing$
(and therefore $\varphi \not= \hmlTrue$).

We prove the three claims of the theorem simultaneously
by induction over the structure of $\varphi$ with arbitrary $p$ and $Q$.
\begin{description}
\item[Case $\varphi = \hmlObs{a} \psi$]
  This means that there exists $p'$ such that $p \step{\ccsInm{a}} p'$ and
  $\psi$ distinguishes $p'$ from $Q' = \{ q' \mid \exists q \in Q \ldotp q \step{\ccsInm{a}} q' \}$.
  The game graph contains an observation move $\attackerPos{p,Q} \gameMoveX{\hmlObs{a}\vphantom{a}} \attackerPos{p',Q'}$
  or $\attackerPos{p,Q}^\notand \gameMoveX{\hmlObs{a}} \attackerPos{p',Q'}$.

  Then we apply the induction hypothesis to $\attackerPos{p',Q'}$ and $\psi$,
  and we conclude that a formula $\psi' \in \varname{strats}[\attackerPos{p',Q'}]$ is found
  such that $\expr(\widehat{\psi'}) \sqsubseteq \expr(\widehat{\psi\vphantom{\psi}})$.
  Therefore, either $\hmlObs{a}\psi' \in \varname{strats}[\attackerPos{p,Q}]$
  (respectively $\hmlObs{a}\psi' \in \varname{strats}[\attackerPos{p,Q}^\notand]$);
  or, if $\hmlObs{a}\psi'$ is pruned away,
  this must have been justified by a formula strictly cheaper than $\hmlObs{a}\psi'$ in $\varname{strats}[\attackerPos{p,Q}]$
  (or $\hmlObs{a}\psi' \in \varname{strats}[\attackerPos{p,Q}^\notand]$),
  and this cheaper formula must have been an observation formula
  according to \refDef{def:formula-pruning}
  (this is needed to prove \refClaim{thm:correctness-general-singletonclaim} if $\relSize{Q} = 1$).
  This proves all three claims.

\item[Case $\varphi = \hmlAnd{i}{I}\psi_i$]
  (Only \refClaim{thm:correctness-general-generalclaim} needs to be proven.)
  Note that $I \not= \varnothing$.
  Assign to every $q \in Q$ a conjunct $\psi_i$ that distinguishes $p$ from $q$;
  we denote this conjunct with $\psi^q$.
  If every $q \in Q$ is assigned the same formula $\psi_{i_0}$
  (e.g.\@ if $Q$ is a singleton),
  apply the induction hypothesis to $\attackerPos{p,Q}$ and $\psi_{i_0}$.

  Otherwise, define a partition of $Q$ as follows:
  \begin{align*}
    \hspace*{\leftmargin}\partition{Q}^- & = \{ \{ q \} \mid \psi^q \text{ is a negation formula} \} \\
    \partition{Q}^+ & = \{ \{ q \mid \psi^q = \psi_i \} \mid \psi_i \text{ is an observation formula, for some } i \in I \} \setminus \{ \varnothing \} \\
    \partition{Q}^{\phantom{+}} & = \partition{Q}^- \cup \partition{Q}^+.
  \end{align*}
  The set $\partition{Q}^-$ contains singleton blocks for every process in $Q$
  that is associated with a negation formula.
  Every block in $\partition{Q}^+$ is associated with an observation formula.
  As we excluded nested conjunctions,
  each $\psi_i$ can only be a negation or an observation formula.
  The game graph contains a conjunct challenge
  $\attackerPos{p, Q} \gameMoveX{\land} \defenderPos{p, \partition{Q}}$.
  We apply the induction hypothesis to the targets of the corresponding conjunct answers:
  \smallskip
  \begin{itemize}
  \item	If some $\{q'\} \in \partition{Q}$ is a singleton,
    we apply the induction hypothesis \refClaim{thm:correctness-general-singletonclaim}
    to $\attackerPos{p,\{q'\}}$ and formula $\psi^{q'}$.
    Then there will be a formula of the same kind as $\psi^{q'}$,
    denoted $\psi'_{\{q'\}} \in \varname{strats}[\attackerPos{p,\{q'\}}]$,
    with $\expr(\widehat{\psi'_{\{q'\}}}) \sqsubseteq \expr(\widehat{\psi^{q'}})$.
  \item	If some $Q' \in \partition{Q}$ is not a singleton,
    then $Q' \in \partition{Q}^+$ and
    $\psi^{q'}$ (for $q' \in Q'$) is an observation formula.
    We apply the induction hypothesis \refClaim{thm:correctness-general-postconjunctionclaim}
    to $\attackerPos{p,Q'}^\notand$ and $\psi^{q'}$.
    Then there will be an observation formula,
    denoted $\psi'_{Q'} \in \varname{strats}[\attackerPos{p,Q'}^\notand]$,
    with $\expr(\widehat{\psi'_{Q'}}) \sqsubseteq \expr(\widehat{\psi^{q'}})$.
  \end{itemize}
  Further note that if $\psi^{q'} = \hmlObs{a}$ is a positive flat branch,
  then the formula found $\psi'_{Q'}$ is also a positive flat branch
  (otherwise it would be more expensive than $\psi^{q'}$).
  Then, the conjunction $\hmlAnd{Q'}{\partition{Q}}\psi'_{Q'}$ satisfies:
  \begin{itemize}
  \item	It contains no more positive branches than $\hmlAnd{i}{I}\psi_i$;%
    \footnote{But the number of positive \emph{flat} branches (with observation height $1$) could have been increased. 
    This is why the first change described in \refRem{rem:new-prices} is necessary.}
  \item	It contains no more positive deep branches than $\hmlAnd{i}{I}\psi_i$;
  \item	For every $\psi'_{Q'}$, there is a $\psi_i$ such that $\smash{\expr(\widehat{\psi'_{Q'}}) \sqsubseteq \expr(\widehat{\psi_i})}$.
  \end{itemize}
  From this, we can conclude that $\expr(\stretchedhat{2.3}{\hmlAnd{Q'}{\partition{Q}}\psi'_{Q'}}) \sqsubseteq \expr(\stretchedhat{1.5}{\hmlAnd{i}{I}\psi_i})$.

  Now either $\hmlAnd{Q'}{\partition{Q}}\psi'_{Q'} \in \varname{strats}[\attackerPos{p,Q}]$,
  or a formula that is strictly cheaper is in $\varname{strats}[\attackerPos{p,Q}]$.
  \refClaim{thm:correctness-general-generalclaim} holds in both cases.

\item[Case $\varphi = \hmlNeg \psi$]
  Note that $\psi$ is an observation formula.
  Assume for now that $Q$ is a singleton $\{ q \}$.
  The game graph contains a negation move $\attackerPos{p,\{q\}} \gameMoveX{\hmlNeg} \attackerPos{q,\{p\}}$.

  We apply the induction hypothesis \refClaim{thm:correctness-general-singletonclaim} to $\attackerPos{q,\{p\}}$ and $\psi$
  and conclude that there is some observation formula $\psi' \in \varname{strats}[\attackerPos{q,\{p\}}]$
  with $\expr(\widehat{\psi'}) \sqsubseteq \expr(\widehat{\psi\vphantom{\psi}})$.
  But then, either $\hmlNeg\psi' \in \varname{strats}[\attackerPos{p,\{q\}}]$,
  or a negation formula strictly cheaper than $\hmlNeg\psi'$ is in $\varname{strats}[\attackerPos{p,\{q\}}]$.
  This proves Claims~\eqref{thm:correctness-general-generalclaim} and~\eqref{thm:correctness-general-singletonclaim} for singleton $Q$.

  Now if $Q$ contains more than one element,
  only \refClaim{thm:correctness-general-generalclaim} needs to be proven.
  The game graph contains a conjunct challenge $\attackerPos{p,Q} \gameMoveX{\land} \defenderPos{p,\{\{q\}\mid q\in Q\}}$,
  and we can apply the above argumentation to the targets $\attackerPos{p,\{q\}}$ of the conjunct answers.
  So, for every $q \in Q$,
  there is some negation formula $\hmlNeg\psi^q \in \varname{strats}[\attackerPos{p,\{q\}}]$
  with $\expr(\widehat{\hmlNeg\psi^q}) \sqsubseteq \expr(\widehat{\hmlNeg\psi})$.
  Consequently, either the conjunction $\hmlAnd{q}{Q}\hmlNeg\psi^q \in \varname{strats}[\attackerPos{p,Q}]$,
  or a formula that is strictly cheaper is in $\varname{strats}[\attackerPos{p,Q}]$.
  As formula prices do not count negative branches in a conjunction,
  $\expr(\stretchedhat{2.1}{\hmlAnd{q}{Q}\hmlNeg\psi^q}) = \expr(\hmlAnd{q}{Q}\hmlNeg\psi^q) \sqsubseteq \expr(\hmlAndS\{\hmlNeg\psi\}) = \expr(\widehat{\hmlNeg\psi})$.
  So the claim on formula prices holds.\footnote{If
    $\varphi = \hmlNeg\hmlObs{a}\hmlNeg\hmlObs{b}$,
    the algorithm might find a formula of the form $\hmlNeg\hmlObs{a}\hmlAndS\{\hmlNeg\hmlObs{b},\hmlNeg\hmlObs{a}\}$.
    The latter formula should not be more expensive than the former.
    This is why the second change in \refRem{rem:new-prices} is necessary.}
  \qedhere
\end{description}
\end{proof}

\noindent
In \refSubsec{subsec:priced-formulas}, we have chosen to describe limits on upper bounds by the depth of negated observations (dimension 6).
An alternative choice could have been to restrict the number of negative deep branches,
but that would make the correctness proof of \refThm{thm:correctness-general} more difficult.
The proof uses the technical property that a formula with more negative branches
is not more expensive than a similar formula with fewer (but more than zero) negative branches.
A formula with one negative branch gives another, technical, reason
why negations under observations should count as implicit conjunctions.

Using the theorem, we can state the correctness of the algorithm easily.

\begin{cor}\label{cor:correctness}
  Assume a spectroscopy game position $\attackerPos{p,Q}$.
  For every language in \refDef{def:ltbts} (and, in relevant cases, $\obs{E}$)
  that distinguishes $p$ from every $q \in Q$,
  Algorithm~\ref{alg:spectroscopy-algo} will find a minimal-price formula in the distinguishing formulas (without double negations or negated conjunctions) in this language.
\end{cor}

\begin{proof}
  Assume that some observation language $\obs{\mathit X}$ distinguishes $p$ from every $q \in Q$.
  In line with \refSubsec{subsec:negated-conjunctions}, we can limit our attention to formulas without double negations or negated conjunctions.
  Then there exists some formula $\varphi \in \obs{\mathit X}$ with minimal formula price
  that distinguishes $p$ from every $q \in Q$.

  Now apply \refThm{thm:correctness-general} to $\varphi$.
  Algorithm~\ref{alg:spectroscopy-algo} will therefore find a formula with the same formula price
  (price-minimality implies that there is no cheaper distinguishing formula than $\varphi$),
  which consequently is in the same distinguishing language $\obs{\mathit X}$.
\end{proof}

\begin{exa}
  Now that we have proven the main correctness theorem,
  we can conclude even more about the minimal-price formulas found in \refExample{exa:game-distinguishing-formulas}.
  As stated there, the processes are not preordered by failure traces nor by impossible futures
  ($\ccsInm{a}.\ccsInm{b} \not\bPreord{FT} \ccsInm{a}\ccsPrefix(\ccsInm{a}\ccsChoice \ccsInm{b}) \ccsChoice \ccsInm{a}\ccsPrefix\ccsInm{b}\ccsPrefix\ccsInm{b} \ccsChoice \ccsInm{a}$ and $\ccsInm{a}.\ccsInm{b} \not\bPreord{IF} \ccsInm{a}\ccsPrefix(\ccsInm{a}\ccsChoice \ccsInm{b}) \ccsChoice \ccsInm{a}\ccsPrefix\ccsInm{b}\ccsPrefix\ccsInm{b} \ccsChoice \ccsInm{a}$),
  but \refCor{cor:correctness} also implies that these are the coarsest ways of telling the processes apart.
  The processes are in every preorder in \refFig{fig:ltbt-spectrum} that is not above one of these two.
  So, we can conclude that they are simulation-preordered and readiness-preordered.
\end{exa}

\subsection{Differences to the Conference Version\label{subsec:difference-tacas}}

The game and algorithm as presented here have an important change when compared to the conference version from TACAS'21~\cite{bisping2021ltbtsTacas}.

In the conference version, each attacker position had only one outgoing conjunction move. \refDef{def:spectroscopy-game} now has one conjunction move from $\attackerPos{p,Q}$ \emph{per partition} of $Q$. This adds a lot of possible game positions and moves. However, this has been necessary to fix an error in the algorithm.

In order to understand the problem, let us reexamine the graph in \refFig{fig:spectro-game-formulas}. The TACAS'21 version would only contain the move from $\attackerPos{b, \{ a \ccsChoice b, b \ccsPrefix b, \ccsStop \}}$ to the finest partition $\defenderPos{b, \{ \{ a \ccsChoice b \}, \{ b \ccsPrefix b \}, \{\ccsStop \} \}}$ but not the one to $\defenderPos{b, \{ \{ a \ccsChoice b \}, \{ b \ccsPrefix b, \ccsStop \} \}}$.
This move, however, was instrumental in finding the failure-trace formula $\hmlObs{a}\hmlAndS\{\hmlObs{b}\hmlNeg\hmlObs{b}, \hmlNeg\hmlObs{a}\}$ in \refExample{exa:game-distinguishing-formulas}.

So, without the correction, the algorithm published originally picks the ready-trace formula $\hmlObs{a}\hmlAndS\{\hmlObs{b}, \hmlObs{b}\hmlNeg\hmlObs{b}, \hmlNeg\hmlObs{a}\}$,
showing that the processes can be distinguished by a ready trace,
namely $\textcolor{darkgreen}{\{ \ccsInm{a} \}} \ccsInm{a} \textcolor{darkgreen}{\{ \ccsInm{b} \}} \ccsInm{b} \textcolor{darkgreen}{\varnothing}$.
While this formula does distinguish $\ccsInm{a}\ccsPrefix\ccsInm{b}$ from $\ccsInm{a}\ccsPrefix(\ccsInm{a}\ccsChoice \ccsInm{b}) \ccsChoice \ccsInm{a}\ccsPrefix\ccsInm{b}\ccsPrefix\ccsInm{b} \ccsChoice \ccsInm{a}$,
and even is semantically equivalent to the failure-trace formula above,
it is not in the cheapest possible language.
Erroneously assuming ready traces to be a minimal language to tell the processes apart leads to the wrong conclusion that the processes would be failure-trace-preordered, which they are not.

The TACAS'21 version lacked a theorem like \refThm{thm:correctness-general}.
The paper, however, was confident that a similar result could be established.
In light of the problem presented here, it is clear that the original algorithm was not complete in the sense of the theorem.

There are several ways of correcting the problem.
We have chosen the least invasive fix by considering all partitions and not just the finest ones.
Together with some slight simplifications of metric and pruning, this allowed us to  provide a completeness proof while staying close to the conference version.

\begin{figure}[p]
\vspace*{6mm} 
  \vspace*{-2mm} 
  \includegraphics[width=\textwidth]{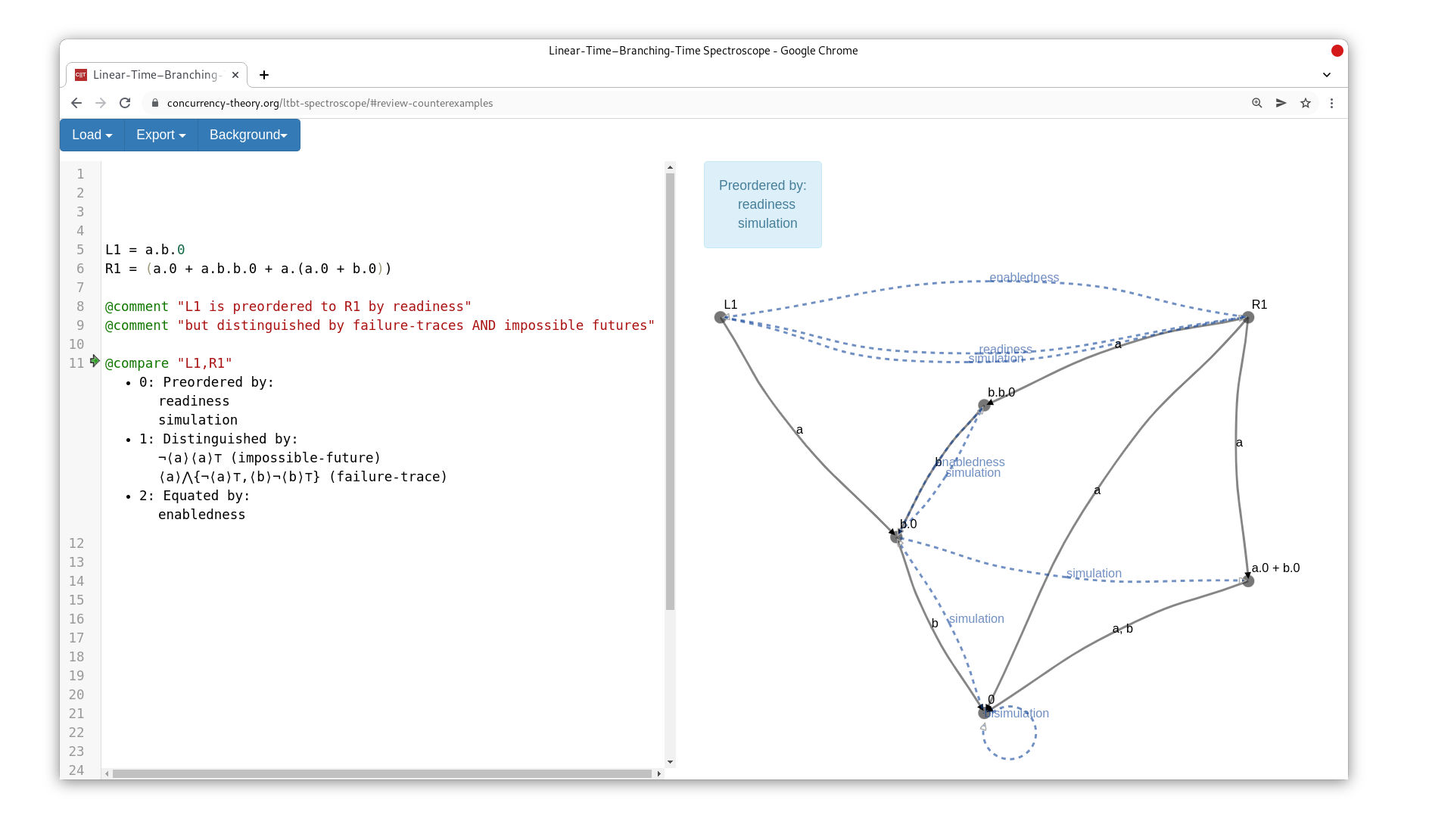}
  \vspace*{-9mm} 
  \caption{Screenshot of a linear-time--branching-time spectroscopy of the processes from \refExample{exa:spec-game-1}.}%
  \label{fig:screenshot}
\vspace*{3mm} 
\end{figure}

\section{Linear-Time--Branching-Time Spectroscopy in the Browser}

There are two implementations of the algorithm described in this article. In \refSubsec{subsec:webtool}, we describe the analysis of transition systems in a web tool built around the algorithm. In \refSubsec{subsec:spectro-browser-game}, we give a short account of a computer game implementation.

\subsection{A Webtool for Equivalence Spectroscopy\label{subsec:webtool}}

We have implemented the game and the generation of minimal distinguishing formulas in the ``Linear-Time--Branching-Time Spectroscope,''
a Scala.js program that can be run in the browser on \url{https://concurrency-theory.org/ltbt-spectroscope/}.

The tool (screenshot in Figure~\ref{fig:screenshot}) consists of a text editor to input finite-state $\ccs$-style processes
(including recursively defined processes)
and a view of the transition system graph.
When queried to compare two processes, the tool yields the cheapest distinguishing $\hml$-formulas it can find for both directions.
From the found formulas, the tool infers the finest fitting preorders for the relevant pairs of processes.
In the process, distinguishing formulas that do not contribute to this question are discarded in order to reduce the noise of the output.

\renewcommand*\arraystretch{1.5}
\begin{table}[p]
  \centering
  \vspace*{3mm} 
  \caption{Formulas found by our implementation for some interesting processes from~\cite{glabbeek2001ltbtsiReport}.}%
  \label{tab:ltbts-formulas}
  \begin{tabular}%
    {>{\raggedright\arraybackslash}p{2.9cm}|%
     >{\raggedright\arraybackslash}p{2.9cm}|%
     >{\raggedright\arraybackslash}p{6.8cm}|%
     >{\raggedright\arraybackslash}l%
  }
  $p$                           & $q$                           & Cheapest distinguishing formulas found & From\\ \midrule
  %
  %
  $\ccsInm{a}\ccsPrefix\ccsInm{b} \ccsChoice \ccsInm{a}$         & $\ccsInm{a}\ccsPrefix\ccsInm{b}$ & $\hmlObsI{a}\hmlNeg\hmlObsI{b} \in \obs{F}$ & p.~13\\
  $\ccsInm{a}\ccsPrefix\ccsInm{b} \ccsChoice \ccsInm{a}\ccsPrefix(\ccsInm{b} \ccsChoice \ccsInm{c})$  &
  $\ccsInm{a}\ccsPrefix(\ccsInm{b}\ccsChoice\ccsInm{c})$                 &
  $\hmlObsI{a}\hmlNeg\hmlObsI{c} \in \obs{F}$ & p.~16                      \\
  $\ccsIdentifier{P_3}$ &
  $\ccsIdentifier{P_4}$ &
  The formulas explained in \refExample{exa:distinguishing-formulas} &
  p.~21
  \\
  $\ccsInm{a}\ccsPrefix\ccsInm{b}\ccsChoice\ccsInm{a}\ccsPrefix(\ccsInm{b}\ccsChoice\ccsInm{c})\ccsChoice\ccsInm{a}\ccsPrefix\ccsInm{c}$ &
  $\ccsInm{a}\ccsPrefix\ccsInm{b}\ccsChoice\ccsInm{a}\ccsPrefix\ccsInm{c}$ &
  $\hmlObsI{a}\hmlAndS \{ \hmlObsI{c},\hmlObsI{b}\} \in \obs{R}\cap\obs{S}$ & p.~24
  \\
  $\ccsInm{a}\ccsPrefix(\ccsInm{b}\ccsChoice\ccsInm{a}\ccsPrefix(\ccsInm{b}\ccsChoice\ccsInm{c}\ccsPrefix\ccsInm{d})\ccsChoice \ccsInm{a}\ccsPrefix\ccsInm{c}\ccsPrefix\ccsInm{e}) \ccsChoice \ccsInm{a}\ccsPrefix(\ccsInm{a}\ccsPrefix\ccsInm{c}\ccsPrefix\ccsInm{d}\ccsChoice\ccsInm{a}\ccsPrefix(\ccsInm{c}\ccsPrefix\ccsInm{e}\ccsChoice\ccsInm{b}))$ &
  $\ccsInm{a}\ccsPrefix(\ccsInm{a}\ccsPrefix(\ccsInm{b}\ccsChoice\ccsInm{c}\ccsPrefix\ccsInm{d})\ccsChoice \ccsInm{a}\ccsPrefix\ccsInm{c}\ccsPrefix\ccsInm{e}) \ccsChoice \ccsInm{a}\ccsPrefix(\ccsInm{a}\ccsPrefix\ccsInm{c}\ccsPrefix\ccsInm{d}\ccsChoice\ccsInm{a}\ccsPrefix(\ccsInm{c}\ccsPrefix\ccsInm{e}\ccsChoice\ccsInm{b}) \ccsChoice \ccsInm{b})$ &
  $\hmlObsI{a}\hmlAndS \{ \hmlObsI{b},\hmlObsI{a} \hmlAndS \{ \hmlObsI{c} \hmlObsI{d},\hmlObsI{b} \} \} \in \obs{RT}\cap\obs{S}$,\newline
  $\hmlObsI{a}\hmlAndS \{ \hmlNeg\hmlObsI{b}, \hmlObsI{a}\hmlAndS \{ \hmlObsI{c} \hmlObsI{d},\hmlNeg\hmlObsI{b} \} \} \in\obs{FT}$&
  p.~27
  \\
  $\ccsInm{a}\ccsPrefix(\ccsInm{b}\ccsPrefix\ccsInm{c}\ccsChoice\ccsInm{b}\ccsPrefix\ccsInm{d})$ &
  $\ccsInm{a}\ccsPrefix\ccsInm{b}\ccsPrefix\ccsInm{c}\ccsChoice\ccsInm{a}\ccsPrefix\ccsInm{b}\ccsPrefix\ccsInm{d}$ &
  $\hmlObsI{a}\hmlAndS \{ \hmlObsI{b}\hmlObsI{c}, \hmlObsI{b}\hmlObsI{d} \} \in\obs{PF}\cap\obs{S}$&
  p.~31
  \\
  $\ccsInm{a}\ccsPrefix\ccsInm{b}\ccsPrefix\ccsInm{c}\ccsChoice\ccsInm{a}\ccsPrefix(\ccsInm{b}\ccsPrefix\ccsInm{c}\ccsChoice\ccsInm{b}\ccsPrefix\ccsInm{d})$ &
  $\ccsInm{a}\ccsPrefix(\ccsInm{b}\ccsPrefix\ccsInm{c}\ccsChoice\ccsInm{b}\ccsPrefix\ccsInm{d})$ &
  $\hmlObsI{a}\hmlNeg\hmlObsI{b}\hmlObsI{d} \in \obs{IF}$ &
  p.~34
  \\
  $\ccsInm{a}\ccsPrefix\ccsInm{b} \ccsChoice \ccsInm{a} \ccsChoice \ccsInm{a}\ccsPrefix\ccsInm{c}$ &
  $\ccsInm{a}\ccsPrefix\ccsInm{b} \ccsChoice \ccsInm{a}\ccsPrefix(\ccsInm{b}\ccsChoice\ccsInm{c}) \ccsChoice \ccsInm{a}\ccsPrefix\ccsInm{c}$ &
  $\hmlObsI{a}\hmlAndS\{\hmlNeg\hmlObsI{b}, \hmlNeg\hmlObsI{c}\} \in \obs{F}$ &
  p.~38
  \\
  $\ccsInm{a}\ccsPrefix\ccsInm{b}\ccsPrefix\ccsInm{c}\ccsChoice\ccsInm{a}\ccsPrefix(\ccsInm{b}\ccsPrefix\ccsInm{c}\ccsChoice\ccsInm{b})$ &
  $\ccsInm{a}\ccsPrefix(\ccsInm{b}\ccsPrefix\ccsInm{c}\ccsChoice\ccsInm{b})$ &
  $\hmlObsI{a}\hmlNeg\hmlObsI{b}\hmlNeg\hmlObsI{c} \in \obs{B}$ &
  p.~42
  \end{tabular}
\end{table}

To ``benchmark'' the quality of the distinguishing formulas,
we have run the algorithm on all the finitary counterexample processes from the report version of ``The Linear-Time--Branching-Time Spectrum''~\cite{glabbeek2001ltbtsiReport}.
Table~\ref{tab:ltbts-formulas} reports the output of our tool, on how to distinguish certain processes.
The results match the (in)equivalences given in~\cite{glabbeek2001ltbtsiReport}.
In some cases, the tool finds slightly better ways of distinction using impossible futures equivalence,
which was not known at the time of the original paper.
All the computed formulas are reasonably small.

For each of the examples (from papers) we have considered,
the browser's capacities sufficed to run the algorithm with pruning in 30 to 250~milliseconds.
This does not mean that one should expect the algorithm to work for systems with thousands of states.
There, the exponentialities of game and formula construction would hit.
However, such big instances would usually stem from preexisting models
where one would very much hope for the designers to already know under which semantics to interpret their model.
The practical applications of our browser tool are more on the research side:
When devising compiler optimizations, encodings, or distributed algorithms,
it can be very handy to fully grasp the equivalence structure of isolated instances.
The Linear-Time--Branching-Time Spectroscope supports this process.

\subsection{A Spectroscopy Browser Game}\label{subsec:spectro-browser-game}

\begin{figure}[t]
  \includegraphics[width=.8\textwidth]{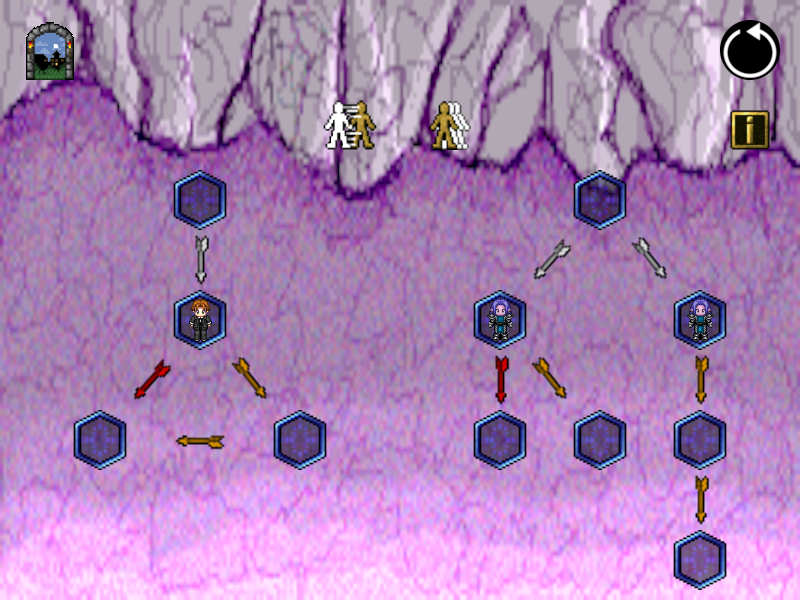}

  \caption{Screenshot of the browser game ``The Spectroscopy Invaders.''}%
  \label{fig:screenshot-game}
\end{figure}

In order to make the spectroscopy game more explainable, Trzeciakiewicz~\cite{trzeciakiewicz2021ltbtBrowserGame} implemented the computer game ``The Spectroscopy Invaders'' where one plays the attacker. A play of the game corresponds to constructing a distinguishing formula. To reach higher scores, one has to construct minimal formulas in the sense of this paper. Under the hood, a TypeScript implementation of Algorithm~\ref{alg:spectroscopy-algo} is used. The game can be played in the browser at \url{https://concurrency-theory.org/ltbt-game/}.

A screenshot of the game is given in \refFig{fig:screenshot-game}.
The player's task is to distinguish $\ccsInm{w}\ccsPrefix(\ccsInm{r} \ccsChoice \ccsInm{y}\ccsPrefix\ccsInm{y})$ (left) from $\ccsInm{w}\ccsPrefix(\ccsInm{r}\ccsChoice\ccsInm{y}) \ccsChoice \ccsInm{w}\ccsPrefix\ccsInm{y}\ccsPrefix\ccsInm{y}$ (right).
The current game position is $\attackerPos{\ccsInm{r} \ccsChoice \ccsInm{y}\ccsPrefix\ccsInm{y}, \{\ccsInm{r}\ccsChoice\ccsInm{y}, \ccsInm{y}\ccsPrefix\ccsInm{y}\}}$,
as indicated by the human figure in the left LTS and the two elves in the right LTS\@.
The player can click the ``conjunction'' button
  \raisebox{-0.3ex}{\begin{tikzpicture}[scale=0.45]
  \draw[very thin,black!70,fill=white] (0.16,0.01) -- ++(-0.2,-0.25) -- ++(0.05,-0.05) --
  ++(0.16,0.2) -- ++(-0.1125,-0.45) -- ++(-0.05,0) -- ++(0,-0.05) --
  ++(0.15,0) -- ++(0.07,0.28) -- ++(0.07,-0.28) -- ++(0.15,0) --
  ++(0,0.05) -- ++(-0.05,0) -- ++(-0.1125,0.45) -- ++(0.16,-0.2) --
  ++(0.05,0.05) -- ++(-0.2,0.25) -- ++(-0.04,0.02) arc[start angle=-74.0380,end angle=200,radius=1mm] -- cycle;
  \draw[very thin,black!85,fill=white] (0.08,0.005) -- ++(-0.2,-0.25) -- ++(0.05,-0.05) --
  ++(0.16,0.2) -- ++(-0.1125,-0.45) -- ++(-0.05,0) -- ++(0,-0.05) --
  ++(0.15,0) -- ++(0.07,0.28) -- ++(0.07,-0.28) -- ++(0.15,0) --
  ++(0,0.05) -- ++(-0.05,0) -- ++(-0.1125,0.45) -- ++(0.16,-0.2) --
  ++(0.05,0.05) -- ++(-0.2,0.25) -- ++(-0.04,0.02) arc[start angle=-74.0380,end angle=200,radius=1mm] -- cycle;
  \draw[very thin,black,fill=yellow!60!black] (0,0) -- ++(-0.2,-0.25) -- ++(0.05,-0.05) --
  ++(0.16,0.2) -- ++(-0.1125,-0.45) -- ++(-0.05,0) -- ++(0,-0.05) --
  ++(0.15,0) -- ++(0.07,0.28) -- ++(0.07,-0.28) -- ++(0.15,0) --
  ++(0,0.05) -- ++(-0.05,0) -- ++(-0.1125,0.45) -- ++(0.16,-0.2) --
  ++(0.05,0.05) -- ++(-0.2,0.25) -- ++(-0.04,0.02) arc[start angle=-74.0380,end angle=254.0380,radius=1mm] -- cycle;
  \end{tikzpicture}
  },
which will split the position into $\attackerPos{\ccsInm{r} \ccsChoice \ccsInm{y}\ccsPrefix\ccsInm{y}, \{\ccsInm{r}\ccsChoice\ccsInm{y}\}}$
and $\attackerPos{\ccsInm{r} \ccsChoice \ccsInm{y}\ccsPrefix\ccsInm{y}, \{\ccsInm{y}\ccsPrefix\ccsInm{y}\}}$.
After that, the player can click on a successor state in the left LTS to indicate an observation move.

Trzeciakiewicz~\cite{trzeciakiewicz2021ltbtBrowserGame} limits the scope to trace, failure, possible-future, simulation, and bisimulation equivalences,
excluding readiness, ready traces, and failure traces.
In effect, the spectrum can be relieved of the issues that come from situations where positive deep branches in a conjunction have to be counted.
This is another way of circumventing the problems discussed in \refSubsec{subsec:difference-tacas}.
In particular, this allows the game mechanics to stay clear of non-trivial partitions,
as they did originally in~\cite{bisping2021ltbtsTacas}.

Moreover, the game is single-player: there is no defender picking a conjunction answer.
Instead, the attacker has to name attacks for every right-hand state. Due to nested conjunctions, the game positions thus actually are sets of $\attackerPos{p, Q}$ tuples.

\section{Related Work and Alternatives}\label{sec:related}

The game and the algorithm presented fill a blank spot in between the following previous directions of work:

\subsection*{Distinguishing Formulas in General.}

Cleaveland~\cite{cleaveland1991} showed how to restore (small but non-minimal) distinguishing formulas for bisimulation equivalence from the execution of a bisimilarity checker based on the splitting of blocks.
He mentioned as possible future work to extend the construction to other notions of the spectrum.
We are not aware of any place where this has previously been done completely.
Korver~\cite{korver1991} extended this work to branching bisimulation;
he needed to extend HML with an until modality that restricts intermediate internal steps.
For example, $\varphi_1 \hmlObs{a} \varphi_2$ means ``the current state can take an $a$-step to a state that satisfies $\varphi_2$, while every intermediate state visited before the $a$-step satisfies $\varphi_1$.''
However, the algorithm in~\cite{korver1991} would not always produce a minimal formula in any sense.
There are related islands like the encoding between CTL and failure traces by Bruda and Zhang~\cite{bz2010modelcheckingCtlFt}.
There is also more recent work like Jasper et al.~\cite{jss2020characteristicInvariantsHML} extending to the generation of characteristic invariant formulas for bisimulation classes,
and like Wi{\ss}mann et al.~\cite{wms2021explainingBehavioralInequivalence} about generating distinguishing formulas for bisimulation in a general coalgebraic setting.
Previous algorithms for bisimulation inequivalence tend to generate formulas that alternate $\hmlObs{a}$ and $[b] \equiv \hmlNeg\hmlObs{b}\hmlNeg$ observations.
Such formulas can not as easily be linked to the spectrum as ours.

\subsection*{Game Characterizations of the Spectrum.}

After Shukla et al.~\cite{shr1995hornsatGames} had shown how to characterize many notions of equivalence by HORNSAT games,
Chen and Deng~\cite{cd2008gameCharacetrizations} presented a hierarchy of games characterizing all the equivalences of the linear-time--branching-time spectrum.
The games from~\cite{cd2008gameCharacetrizations} allow word moves and thus are infinite already for finite transition systems with cycles.
Therefore, they cannot be applied as easily as ours in algorithms.
Constructing distinguishing formulas from attacker strategies of these games would be less convenient than in our solution.
Their parametric approach is comparable to fixing maximal price budgets \emph{ex ante.}
Our on-the-fly picking of minimal prices provides more flexibility.

\subsection*{Using Game Characterizations for Distinguishing Formulas.}

There is recent work by Mika-Michalski et al.~\cite{km2020nonBisimCoalgebraic} on constructing distinguishing formulas using games in a more abstract coalgebraic setting focussed on the absence of bisimulation. The game and formula generation there, however, cannot easily be adapted for our purpose of performing a \emph{spectroscopy} also for weaker notions.

\subsection*{Generalizations of the Spectrum.}

As a by-product of our work, \refSubsec{subsec:priced-formulas} has recharted the linear-time--branching-time spectrum in a way that corresponds nicely with our algorithm. Each coordinate in the price lattice can be seen as characterizing a conceivable notion of equivalence. Thus, this characterization generalizes (a part of) the~\cite{glabbeek2001ltbtsiReport}-spectrum. For instance, we have not considered the revivals semantics~\cite{roscoe2009revivalsHierarchy}, but it can easily be characterized by $e_\mathrm{R} \sqcap e_\mathrm{FT} = (\infty, 1, 0, 1, 1, 1)$. So let us briefly mention important generalizations of the spectrum.

For instance, Lange et al.~\cite{lange2014modelCheckingProcEqs} introduce a \emph{higher-order dyadic $\mu$-calculus} where each single formula characterizes a notion of equivalence. This, too, allows a more unified computational treatment of the spectrum's members and many conceivable siblings.

The most complete \emph{observational generalization} of~\cite{glabbeek2001ltbtsiReport} has been provided by de Frutos Escrig et al.~\cite{erph2013unfyingLTBTS}. Their main dimensions are five layers of simulation, which correspond to classes of localizable annotations on an observaion tree (viz.:\@ no local observations, deadlock detection, enabled actions, enabled traces, enabled trees), and in which way to compare these annotations. The second main dimension actually is a combination of 5 to 9 sub-dimensions: (1) whether to compare only along one linear path, deterministic branching paths or all branching paths, whether to check for (2) superset and/or (3) subset relationship between annotations, and whether to do (each of) this (4) along the paths and/or (5) at the leafs of observation trees. Their lattice cannot only cover interesting notions of equivalence unreachable by ours (for instance impossible-futures-along-a-trace), but is also nicely linked to a unified system of axiomatic characterizations. Still, their characterization is less general than our price lattice in other regards. For instance, it cannot cover counting equivalences like $i$-step bisimilarity, $(i,\infty,\infty,\infty,\infty,\infty)$ or $n$-nested similarity.

\subsection*{Apartness and Simulation Distances}

Published around the same time as the conference version of the present paper, Geuvers and Jacobs~\cite{geuversJacobs2021apartness} discussed \emph{apartness} as the inductively defined dual of bisimulation. This is closely related to the approach of the present paper. We also tackle the hierarchy of equivalence problems by rather addressing the dual question of how easy it is to tell two processes apart. In this view, the expressiveness prices output by our algorithm provide information on ``how far apart'' two models are.

Another view on ``how far apart'' models are is provided by \emph{(bi)simulation distances} as researched by \v{C}ern{\'y} et al.~\cite{cerny2012simulationDistances}, and by Romero Hern{\'a}ndez and de Frutos Escrig~\cite{hernandezEscrig2012definingDistancesForAll}.
In this line of work, distances quantify how many and how relevant changes need to be made to a process in order to make it equivalent to another process with respect to a fixed semantics.
This is orthogonal to what we are doing:
Bisimulation distances are about ``how much cheating'' in a (bi)simulation game the defender would need to win.
Our multidimensional discrete expressiveness prices are about ``how much semantics'' the attacker needs to win in a bisimulation game without cheating.

\subsection*{Alternatives}

One can also find the finest notion of equivalence between two states
by \emph{gradually minimizing} the transition system with ever coarser equivalences from bisimulation to trace equivalence
until the states are conflated (possibly also trying branches of the spectrum).
Within a big tool suite of highly optimized algorithms this should be quite efficient.
We preferred the game approach, because it can uniformly be extended to the whole spectrum
and also has the big upside of explaining the in-equivalences by distinguishing formulas.

A variation of our approach, which we have already tried,
is to run the formula search on a \emph{directed acyclic subgraph} of the winning strategy graph.
For our purpose of finding most fitting equivalences,
DAG-ification may preclude the algorithm from finding the right formulas.
On the other hand, if one for instance is mainly interested in a short distinguishing formula with low depth,
one can speed up the process with DAG-ification by the order of remaining game rounds.

\subsection*{Optimizations}

The algorithm in this article has more than exponential time complexity.
To improve on that, one could reduce the number of conjunct challenges generated as follows.
First one only explores the finest conjunct challenge move $\attackerPos{p,Q} \gameMoveX{\land} \defenderPos{p,\{\{q\}\mid q \in Q\}}$,
and if that succeeds and leads to a formula with multiple positive branches,
one only needs to explore one or two additional conjunct challenges to find out whether there is a formula with a single positive (deep) branch.
Only formulas with one positive (deep) branch may be in a smaller language than the previously found formula.
This faster variant would still generate distinguishing formulas from the cheapest languages
(i.e.\@ they solve Problem~\ref{prob:abstract-obs-problem})
but these formulas are not always the cheapest ones in themselves
(i.e.\@ they do not solve Problem~\ref{prob:cheapest-distinction}).
This would require to interleave the game graph construction
(line~\ref{algo:line:construct-spectroscopy-game} of our Algorithm~\ref{alg:spectroscopy-algo})
with the generation of distinguishing formulas (lines~\ref{algo:line:begin-distinguish-formula}--\ref{algo:line:end-distinguish-formula}).

There are some other avenues for small optimizations:
For instance, the pruning of \refDef{def:formula-pruning} can be changed to take game positions into account.
If the right-hand state set $Q$ of $\attackerPos{p,Q}$ has more than one element, dominated formulas can be pruned regardless of their kind.
Also, the game construction can be sped up a bit by stopping the construction of attacker moves in positions that cannot be winning for the attacker.
In particular, this is the case for $\attackerPos{p,Q}$ if $p \in Q$.
The downside of this is that it becomes a little more difficult to retrieve a proper bisimulation relation in positions where the attacker does not win.

\section{Conclusion}%
\label{sec:conclusion}

In this paper, we have established a convenient way of finding distinguishing formulas that use a minimal amount of expressiveness.

System analysis tools can employ the algorithm to tell their users in more detail \emph{how equivalent} two process models are. While the generic approach is costly, instantiations to more specific, symbolic, compositional, on-the-fly or depth-bounded settings may enable wider applications. There are also some algorithmic tricks (like building the concrete formulas only after having found the price bounds and heuristics in handling the game graph) we have not explored in this paper.

More clever ways of characterizing the spectrum, the formula prices, and instances of admissible pruning might further improve the applicability of the approach.

So far, we have only looked at \emph{strong} notions of equivalence~\cite{glabbeek1990ltbt1}.
It is an interesting question how to extend our algorithm,
so it also deals with \emph{weak} notions of equivalence~\cite{glabbeek1993ltbt}.
These equivalences abstract over $\tau$-actions representing ``internal activity''
and correspond to observation languages with a special temporal $\hmlObs{\epsilon}$-observation (cf.~\cite{gfm2020congruenceOperator}).
Such an extension would generalize work on weak game characterizations such as de Frutos Escrig et al.'s~\cite{ekw2017gamesBisimAbstraction}
and our own~\cite{bn2019coupledsimTacas,bnp2020coupledsim30,bm2021contrasimilarity,SmetsersMJ16}.

A roadblock to expanding the approach to the weak spectrum lies in the fact that,
with the conventional modal characterizations of the weak equivalences, the observation languages are not closed.
Subformulas would need to be more expensive than their parent formulas.
For instance, $\hmlObs{\epsilon}\hmlObs{a}\hmlObs{\epsilon}\hmlTrue$ is a valid weak trace observation
but its subformula $\hmlObs{a}\hmlObs{\epsilon}\hmlTrue$ is not.
In order to apply our algorithm, one would need a different modal language, perhaps similar to Korver's until operator~\cite{korver1991}.

The vision is to arrive at \emph{one} certifying algorithm that can yield finest equivalences and cheapest distinguishing formulas as witnesses for the whole spectrum.

\subsection*{Acknowledgments}

We are thankful to members of the MTV research group and the reviewers of TACAS'21 and LMCS for lots of helpful comments.

\subsection*{Data availability}

The source code git repository of our implementation is on GitHub and can be accessed via \url{https://concurrency-theory.org/ltbt-spectroscope/code/}. The accompanying source code is archived on Zenodo~\cite{bisping2022ltbtsZenodo2}. The deprecated TACAS'21 version remains available there as well~\cite{bisping2021ltbtsZenodo}.

\bibliographystyle{alphaurl}

\bibliography{similarities}

\end{document}